\documentclass[a4paper,11pt,onecolumn,accepted=2026-03-25]{quantumarticle}
\pdfoutput=1 
\usepackage{amsmath,amsxtra,amssymb,amsthm,amsfonts,bm}

\usepackage{graphicx}
\usepackage[font=small,labelfont=bf,
   justification=justified,
   format=plain]{subcaption}
\usepackage[font=small,labelfont=bf,
   justification=justified,
   format=plain]{caption}
\usepackage{svg}
\usepackage{comment}
\usepackage{tikz}
\usepackage{circuitikz}
\usepackage{bbm}
\usepackage{enumitem}
\usepackage{amsfonts}

\usepackage{bbm}
\usepackage{url}
\usepackage{stmaryrd}

\usepackage{braket}
\usepackage{pgfplots}
\usepackage{upgreek}
\usepackage{dsfont}
\usepackage{makecell}
\usepackage{fancyref} 
\usepackage{algorithm}
\usepackage{algpseudocode}
\usepackage[utf8]{inputenc}

\usepackage{comment}

\usepackage{float}
\newfloat{algorithm}{t}{lop}

\newcommand{\abs}[1]{\left| #1 \right|}

\usepackage{hyperref}

\theoremstyle{plain} 

\newtheorem{theorem}{Theorem}[section]
\newtheorem{alg}{Algorithm}[section]

\newtheorem{lem}{Lemma}[section]

\newtheorem{prop}{Proposition}[section]

\newtheorem{defi}{Definition}[section]

\theoremstyle{remark}

\newcommand*{\fancyrefthmlabelprefix}{thm}
\newcommand*{\fancyreflemlabelprefix}{lem}
\newcommand*{\fancyrefcorlabelprefix}{cor}
\newcommand*{\fancyrefdefilabelprefix}{defi}
\frefformat{plain}{\fancyreflemlabelprefix}{lemma\fancyrefdefaultspacing#1}
\Frefformat{plain}{\fancyreflemlabelprefix}{Lemma\fancyrefdefaultspacing#1}
\frefformat{plain}{\fancyrefthmlabelprefix}{theorem\fancyrefdefaultspacing#1}
\Frefformat{plain}{\fancyrefthmlabelprefix}{Theorem\fancyrefdefaultspacing#1}
\frefformat{plain}{\fancyrefcorlabelprefix}{corollary\fancyrefdefaultspacing#1}
\Frefformat{plain}{\fancyrefcorlabelprefix}{Corollary\fancyrefdefaultspacing#1}
\frefformat{plain}{\fancyrefdefilabelprefix}{definition\fancyrefdefaultspacing#1}
\Frefformat{plain}{\fancyrefdefilabelprefix}{Definition\fancyrefdefaultspacing#1}
\newcommand*{\fancyrefalglabelprefix}{alg}
\newcommand*{\frefalgname}{algorithm}
\newcommand*{\Frefalgname}{Algorithm}
\frefformat{plain}{\fancyrefalglabelprefix}{%
  \frefalgname\fancyrefdefaultspacing#1%
}%
\Frefformat{plain}{\fancyrefalglabelprefix}{%
  \Frefalgname\fancyrefdefaultspacing#1%
}%

\newcommand*{\fancyrefapplabelprefix}{app}
\newcommand*{\frefappname}{appendix}
\newcommand*{\Frefappname}{Appendix}
\frefformat{plain}{\fancyrefapplabelprefix}{%
  \frefappname\fancyrefdefaultspacing#1%
}%
\Frefformat{plain}{\fancyrefapplabelprefix}{%
  \Frefappname\fancyrefdefaultspacing#1%
}%

\definecolor{Green}{HTML}{00AD69}  

\def\beq{\begin{equation}}
\def\eeq{\end{equation}}
\def\bq{\begin{quote}}
\def\eq{\end{quote}}
\def\ben{\begin{enumerate}}
\def\een{\end{enumerate}}

\def\bit{\begin{itemize}}
\def\eit{\end{itemize}}

\def\l|{\left|}
\def\r|{\right|}

\newcommand\R{\mathbbm{R}}

\newcommand{\tr}{{\operatorname{tr}}}

\newcommand{\norm}[1]{\left\|#1\right\|}

\newcommand{\sign}{\text{sign}}

\newcommand{\Tr}{{\operatorname{tr}}}

\begin{document}

\author{Victor Martinez}
\email[]{victor.l.martinez@ibm.com}
\affiliation{IBM France, Avenue de l'Europe, 92275 Bois-Colombes, France}
\affiliation{Univ Lyon, ENS Lyon, UCBL, Inria, LIP, F-69342, Lyon Cedex 07, France}
\author{Omar Fawzi}
\affiliation{Univ Lyon, ENS Lyon, UCBL, Inria, LIP, F-69342, Lyon Cedex 07, France}

\author{Daniel Stilck Fran\c{c}a}

\affiliation{Univ Lyon, ENS Lyon, UCBL, Inria, LIP, F-69342, Lyon Cedex 07, France}
\affiliation{Department of Mathematical Sciences, University of Copenhagen, Universitetsparken 5, 2100 Copenhagen, Denmark}

\title{Sampling (noisy) quantum circuits through randomized rounding}

\begin{abstract}

The present era of quantum processors with hundreds to thousands of noisy qubits has sparked interest in understanding the computational power of these devices and how to leverage it to solve practically relevant problems. For applications that require estimating expectation values of observables the community developed a good understanding of how to simulate them classically and denoise them. Certain applications, like combinatorial optimization, however demand more than expectation values: the bit-strings themselves encode the candidate solutions. While recent impossibility and threshold results indicate that noisy samples alone rarely beat classical heuristics, we still lack classical methods to replicate those noisy samples beyond the setting of random quantum circuits.

Focusing on problems whose objective depends only on two-body correlations such as Max-Cut, we show that Gaussian randomized rounding in the spirit of Goemans-Williamson applied to the circuit's two-qubit marginals produces a distribution whose expected cost is provably close to that of the noisy quantum device. For instance, for Max-Cut problems we show that for any depth-D circuit affected by local depolarizing noise p, our sampler achieves a recovery ratio $1-O[(1-p)^D]$, giving ways to efficiently sample from a distribution that behaves similarly to the noisy circuit for the problem at hand. Beyond theory we run large-scale simulations and experiments on IBMQ hardware, confirming that the rounded samples faithfully reproduce the full energy distribution, and we show similar behaviour under other various noise models.
  
Our results supply a simple classical surrogate for sampling noisy optimization circuits, clarify the realistic power of near-term hardware for combinatorial tasks, and provide a quantitative benchmark for future error-mitigated or fault-tolerant demonstrations of quantum advantage.

\end{abstract}
\maketitle

\section{Introduction}
Quantum computing is now at a pivotal stage, with its potential exceeding the limits of classical computing and posing challenges for simulating quantum devices through classical means \cite{Arute2019}. In light of this, the quantum computing community is working diligently to ascertain whether near-term quantum computers, despite inherent noise and the lack of effective quantum error correction, can surpass classical computers in tackling practical problems.

A significant focus of recent research has been the application of near-term quantum devices to solve optimization problems in various areas, as well as approximating the ground-state energy of relevant physical Hamiltonians \cite{Kokail2019,McArdle2019,Amaro_2022, Yoshioka2025}. Contrasting with quantum advantage tests based on random circuit sampling, optimization offers tangible, real-world applications. Furthermore, it is relatively simple to compare the efficacy of a quantum device to a classical computer by evaluating which one achieves the lowest energy prediction. However, the inherent noise in current devices can significantly affect the accuracy and reliability of the results, making it difficult to obtain optimal solutions for complex combinatorial optimization problems. 

Two primary approaches exist for mitigating the effects of noise: error correction and error mitigation. Error correction represents a long-term, theoretically robust solution that aims to preserve quantum coherence and computation fidelity through the use of redundant quantum encoding and fault-tolerant protocols. However, this approach requires substantial overhead in terms of qubits and gates, making it currently impractical for near-term quantum devices. The resource-intensive nature of error correction  \cite{Terhal2015}, combined with the limitations of present-day quantum hardware, has pushed its implementation into a future where large-scale, fault-tolerant quantum computers may become feasible.

Error mitigation, on the other hand, provides a more practical, short-term strategy for addressing noise. By leveraging classical post-processing techniques on the results of quantum computations, error mitigation allows us to approximate the outcomes of a noiseless quantum circuit without requiring significant additional quantum resources \cite{Cai2023}. Although error mitigation has shown promise in reducing the effects of noise \cite{Yu2023,Shtanko2023,Farrell2023}, it is inherently limited in scalability \cite{Takagi2022,quek2023exponentially}. However, error mitigation does not provide us with samples from the quantum circuit, only estimates of noiseless expecation values. Crucially, the actual bitstrings sampled from the output of the noiseless circuit are necessary to obtain an assignment to the optimization problem.
Consequently, if we aim to obtain an actual assignment, we are constrained to using the samples generated by the noisy quantum device. This raises a critical question: \textit{can samples with similar performance be obtained at a smaller cost, e.g., without having to run the quantum circuit?}

Under depolarizing noise, a common noise model, we are beginning to develop a clearer understanding of the limitations quantum devices face. Previous work has established that there exists a constant noise threshold beyond which simple classical algorithms are expected to outperform quantum approaches \cite{StilckFrana2021}. This insight suggests that, as the noise level increases, quantum devices lose their computational advantage in solving combinatorial optimization problems. Nonetheless, sampling noisy quantum circuits directly remains computationally intractable in most cases \cite{Liu2024}. The behavior of quantum circuits at depths significantly smaller than this noise threshold remains poorly understood, and the situation becomes even more complex when considering nonunital noise models, such as amplitude damping \cite{benor2013, Fefferman2023}.

We address this outstanding issue by observing that for many optimization problems such as MaxCut or Ising problems, only two-body correlations are relevant to the problem \cite{Lucas2014}. We show that by employing randomized rounding techniques, it is possible to mimic the behavior of the quantum circuit designed to tackle the optimization problem. Specifically, we present a very simple algorithm with performance guarantees that generates samples from a distribution designed to replicate the noisy circuit, given only access to two-body expectation values. We prove various guarantees both \textit{a priori} and \textit{a posteriori}, which allow us to determine if sampling from a given noisy QAOA will be advantageous in solving the problem. Notably, even though our theoretical guarantees are focused on the recovery ratio, our algorithm appears to capture the entire distribution of values remarkably well in practice. Moreover, our method does not require the depth or noise level to be greater than a specific threshold, but its performance improves as the amount of noise increases, making it a versatile and robust approach for various scenarios. Thus, we provide a simple yet powerful algorithm that provably works and performs well in practice, offering a reliable way to benchmark noisy quantum circuits. We summarize the framework we will be working in along with our main contributions in Figure \ref{fig:framework}.

\begin{figure}[H]
\centering
\resizebox{1\textwidth}{!}{%
\begin{circuitikz}
\tikzstyle{every node}=[font=\tiny]
\draw [rounded corners = 4.8] (6.25,11.25) rectangle (3.5,6.25);
\draw [rounded corners = 4.8] (9,6.25) rectangle (11.75,11.25);
\draw [->, >=Stealth] (6.25,10.5) -- (9,10.5);
\draw [->, >=Stealth] (6.25,8.75) -- (9,8.75);
\draw [->, >=Stealth] (6.25,7) -- (9,7);
\node [font=\tiny] at (7.5,11) {\textbf{Quantum}};
\node [font=\tiny] at (7.5,7.5) {\textbf{Classical noisy}};
\node [font=\tiny] at (7.5,9.25) {\textbf{Exact}};
\node [font=\tiny] at (7.6,7.25) {\textbf{expectation values}};
\node [font=\tiny] at (7.55,6.7) {Single qubit noise};\node [font=\tiny] at (7.55,6.4) {\cite{fontana2023classical,Martinez2025}};
\draw [rounded corners = 6.0] (-0.75,10.75) rectangle (1.75,8.25);
\draw [rounded corners = 6.0] (16.25,10.75) rectangle (13.75,8.25);
\draw [](0.5,8.25) to[short] (0.5,5);
\draw [->, >=Stealth] (0.5,5) -- (13.55,5);
\draw [->, >=Stealth] (11.75,9.75) -- (13.75,9.75);
\draw [->, >=Stealth] (11.75,6.75) -- (13.55,6.75);
\draw [->, >=Stealth] (1.75,9.5) -- (3.5,9.5);
\node [font=\footnotesize] at (15,5) {\textbf{A priori bounds}};
\node [font=\footnotesize] at (15.25,6.75) {\textbf{A posteriori bounds}};
\node [font=\tiny] at (15.25,6.4) {Max-Cut: Prop.\ref{propMCnoiseless}};
\node [font=\tiny] at (15.25,6.15) {QUBO: Th. \ref{theoremNoisyQUBO}};
\node [font=\tiny] at (15,4.65) {Max-Cut: Th. \ref{theoremNoisyMaxcut}};
\node [font=\tiny] at (15,4.4) {QUBO: Th. \ref{theoremNoisyQUBO}};
\node [font=\tiny] at (12.75,10.25) {\textbf{Sampling}};
\node [font=\tiny] at (12.75,10) {\textbf{Algorithm}};
\node [font=\tiny] at (12.75,9.55) {\ref{samplingAbs}};
\node [font=\footnotesize] at (4.85,10.25) {Circuit $\mathcal{C}$};
\node [font=\tiny] at (7.5,10.75) {\textbf{error mitigation}};
\node [font=\tiny] at (7.6,10.2) {Small depths};
\node [font=\tiny] at (7.6,9) {\textbf{expectation values}};
\node [font=\tiny] at (7.6,8.45) {Small depths};
\node [font=\footnotesize] at (10.35,10) {Expectation};
\node [font=\footnotesize] at (10.35,9.5) {values};
\draw [ fill={rgb,255:red,251; green,236; blue,176} , rounded corners = 3.0, ] (9.25,8.5) rectangle (11.5,7.5);
\node [font=\scriptsize] at (10.37,8) {$\langle O_i \rangle=\tr([\mathcal{C}]O_i)$};
\draw [short] (0.5,8.75) -- (1.25,9.25);
\draw [short] (1.25,9.25) -- (0.75,10);
\draw [short] (0.5,8.75) -- (-0.25,9.5);
\draw [short] (-0.25,9.5) -- (0.75,10);
\draw [short] (0.75,10) -- (0,10.25);
\draw [short] (0,10.25) -- (-0.25,9.5);
\node [font=\footnotesize] at (15.05,11) {\texttt{00101}};
\draw [ fill={rgb,255:red,145; green,145; blue,145} ] (-0.25,9.5) circle (0.15cm);
\draw [ fill={rgb,255:red,145; green,145; blue,145} ] (0,10.25) circle (0.15cm);
\draw [ fill={rgb,255:red,145; green,145; blue,145} ] (0.75,10) circle (0.15cm);
\draw [ fill={rgb,255:red,145; green,145; blue,145} ] (0.5,8.75) circle (0.15cm);
\draw [ fill={rgb,255:red,145; green,145; blue,145} ] (1.25,9.25) circle (0.15cm);
\draw [short] (15,8.75) -- (15.75,9.25);
\draw [short] (15.75,9.25) -- (15.25,10);
\draw [short] (15,8.75) -- (14.25,9.5);
\draw [short] (14.25,9.5) -- (15.25,10);
\draw [short] (15.25,10) -- (14.5,10.25);
\draw [short] (14.5,10.25) -- (14.25,9.5);
\draw [ fill={rgb,255:red,200; green,200; blue,255} , line width=0.2pt ] (14.25,9.5) circle (0.15cm);
\draw [ fill={rgb,255:red,254; green,188; blue,188} , line width=0.2pt ] (14.5,10.25) circle (0.15cm);
\draw [ fill={rgb,255:red,254; green,188; blue,188} , line width=0.2pt ] (15.25,10) circle (0.15cm);
\draw [ fill={rgb,255:red,254; green,188; blue,188} , line width=0.2pt ] (15,8.75) circle (0.15cm);
\draw [ fill={rgb,255:red,200; green,200; blue,255} , line width=0.2pt ] (15.75,9.25) circle (0.15cm);
\draw [](3.75,9) to[short] (4,9);
\draw [](3.75,8.5) to[short] (4,8.5);
\draw [](3.75,8) to[short] (4,8);
\draw [ fill={rgb,255:red,200; green,200; blue,255} , line width=0.2pt ] (4,7.65) rectangle (4.25,6.85);
\draw [ fill={rgb,255:red,200; green,200; blue,255} , line width=0.2pt ] (4.25,7.8) rectangle (4,8.2);
\draw [](3.75,7.5) to[short] (4,7.5);
\draw [](3.75,7) to[short] (4,7);
\draw [ fill={rgb,255:red,200; green,200; blue,255} , line width=0.2pt ] (4,9.15) rectangle (4.25,8.35);
\draw [](4.25,9) to[short] (4.5,9);
\draw [](4.25,8.5) to[short] (4.5,8.5);
\draw [](4.25,8) to[short] (4.5,8);
\draw [](4.25,7.5) to[short] (4.5,7.5);
\draw [](4.25,7) to[short] (4.5,7);
\draw [ fill={rgb,255:red,200; green,200; blue,255} , line width=0.2pt ] (4.5,8.8) rectangle (4.75,9.2);
\draw [ fill={rgb,255:red,200; green,200; blue,255} , line width=0.2pt ] (4.5,8.65) rectangle (4.75,7.85);
\draw [ fill={rgb,255:red,200; green,200; blue,255} , line width=0.2pt ] (4.5,7.7) rectangle (4.75,7.3);
\draw [ fill={rgb,255:red,200; green,200; blue,255} , line width=0.2pt ] (4.75,6.8) rectangle (4.5,7.2);
\draw [](4.75,7.5) to[short] (5,7.5);
\draw [](4.75,7) to[short] (5,7);
\draw [](4.75,8) to[short] (5,8);
\draw [](4.75,8.5) to[short] (5,8.5);
\draw [](4.75,9) to[short] (5,9);
\draw [ fill={rgb,255:red,200; green,200; blue,255} , line width=0.2pt ] (5,8.15) rectangle (5.25,7.35);
\draw [ fill={rgb,255:red,200; green,200; blue,255} , line width=0.2pt ] (5.25,6.8) rectangle (5,7.2);
\draw [ fill={rgb,255:red,200; green,200; blue,255} , line width=0.2pt ] (5.25,8.35) rectangle (5,9.15);
\draw [](5.25,7) to[short] (5.5,7);
\draw [](5.25,7.5) to[short] (5.5,7.5);
\draw [](5.25,8) to[short] (5.5,8);
\draw [](5.25,8.5) to[short] (5.5,8.5);
\draw [](5.25,9) to[short] (5.5,9);
\draw [](5.75,9) to[short] (6,9);
\draw [ fill={rgb,255:red,200; green,200; blue,255} , line width=0.2pt ] (5.5,9.15) rectangle (5.75,8.75);
\draw [ fill={rgb,255:red,200; green,200; blue,255} , line width=0.2pt ] (5.5,8.65) rectangle (5.75,7.85);
\draw [ fill={rgb,255:red,200; green,200; blue,255} , line width=0.2pt ] (5.5,7.65) rectangle (5.75,6.85);
\draw [](5.75,8.5) to[short] (6,8.5);
\draw [](5.75,8) to[short] (6,8);
\draw [](5.75,7.5) to[short] (6,7.5);
\draw [](5.75,7) to[short] (6,7);
\end{circuitikz}
}
\caption{Representation of the framework and main contributions. An optimization instance is given in the form of a graph, and $\mathcal{C}$ is a quantum circuit to solve this instance. Expectation values are extracted from the quantum circuit, and sampling algorithms allow to recover samples from these along with guarantees on the quality of the samples.}
\label{fig:framework}
\end{figure}
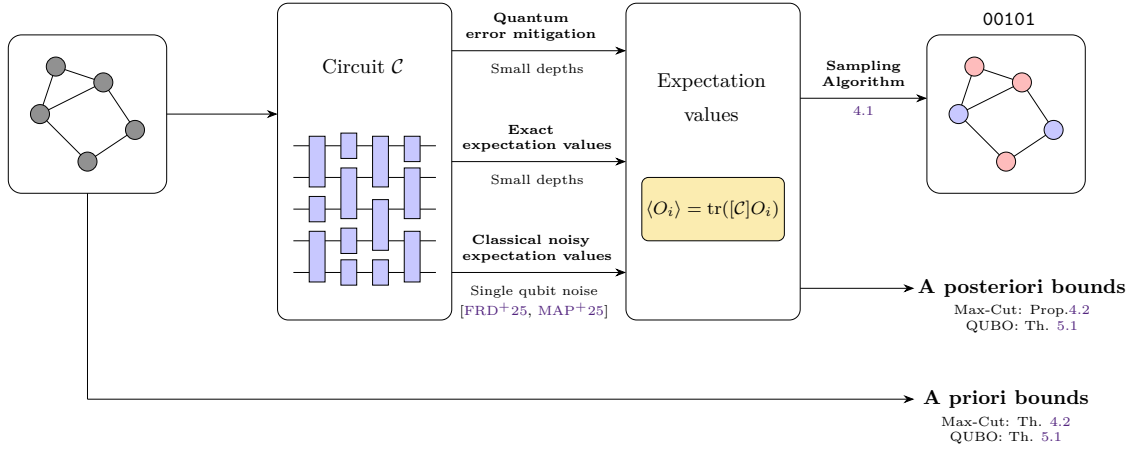

The paper is structured as follows: In Section \ref{sec:main_results}, we provide an overview of our main results, summarizing the performance guarantees of our classical surrogate compared to (noisy) quantum circuits. Section \ref{sec:background} introduces the necessary notations and the preliminary framework regarding combinatorial optimization and noise models. In Section \ref{sectionMC}, we formally define our sampling algorithm based on Gaussian rounding and we prove its efficacy for Max-Cut problems by deriving bounds on the expected recovery ratio. We extend this analysis to general QUBO formulations in Section \ref{sectionQUBO}. Section \ref{sectionCov} details the end-to-end implementation of our sampling scheme, including the computation of the required expectation values. Finally, Section \ref{sectionNumerics} provides numerical evidence, benchmarking our method against IBMQ hardware and simulations.

\section{Summary of main results}\label{sec:main_results}
\subsection{Problem setting and objective}

In this paper, we are interested in combinatorial optimization problems characterized by a local cost function $C$ of the form $C(z)=z^TAz$. The vector $z=(z_1,\ldots,z_n)$ represents a bitstring in $\{-1,+1\}^n$ and $A\in \mathbb{R}^{n\times n}$ a problem-specific matrix. This formulation encapsulates various known combinatorial optimization problems based on the structure of the cost matrix $A$ such as Max-Cut, the Quadratic Unconstrained Binary Optimization (QUBO) problem, and many others \cite{Lucas2014}. These problems are often described by graphs $\mathcal{G}=(V,E)$, where the set of nodes $V$ represents the variable $z$ and an edge $(i,j)\in E$ between node $i$ and node $j$ exists if the coefficient $A_{ij}$ is non-zero. 

Many classes of such problems are NP-hard \cite{Karp1972}, and finding better heuristics to solve them is of particular industrial interest. These problems can be mapped to a diagonal Hamiltonian $H$ such that solving the combinatorial optimization problem is equivalent to finding the ground state of the Hamiltonian. The Hamiltonian to minimize exposes the same locality as the original cost function $C$ and can be written in a general form $H=-\sum_{i,j}J_{ij}Z_iZ_j-\sum_{i}h_iZ_i$, where the values of $J$ and $h$ depend on $C$. In the following, we will only consider the case where $h_i=0$, although our sampling algorithm could also handle problems for which $h_i \neq 0$.

Consider a circuit $\mathcal{C}$, which aims at finding the ground state of a Hamiltonian $H$. Upon running the circuit on an initial state, the circuit outputs a state described by a density matrix $\rho$. In the following, we denote $[\mathcal{C}]=\rho$ the quantum state produced by running circuit $\mathcal{C}$ on any input state. We will also consider noisy versions of the circuit.

That is, for a qubit channel $\mathcal{N}$ (which models noise), we assume that at every layer, $\mathcal{N}$ is applied to all qubits. Under this noise model, we denote similarly $[\mathcal{C}]_\mathcal{N}$ the quantum state produced by the noisy circuit. Once the state is prepared by the (noisy) quantum circuit, we will denote by $\mathcal{M}([\mathcal{C}])$ the distribution over bitstrings $\{-1,1\}^n$ we obtain when we measure it in the computational basis. The hope is that by sampling from such a distribution we obtain a string $z$ that has low cost for the function $C$ (or equivalently low energy for $H$). This description encapsulates many quantum algorithms designed for combinatorial optimization, such as QAOA \cite{farhi2014quantum} or VQE \cite{Peruzzo2014}. Since our Hamiltonian only consists of terms acting on one and two qubits at a time, the expectation value of the cost function is fully determined by the the expectation values of $[\mathcal{C}]$ (or $[\mathcal{C}]_{\mathcal{N}}$ in the noisy case) on at most two sites. A natural question is: \textit{Given the two-sites expectation values, can we obtain a $z \in \{-1,1\}^n$ that attains a value comparable to the output of the quantum circuit?}

We denote by $\mathbb{E}_{z \sim \mathcal{M}([\mathcal{C}])}[C(z)]$ the expected cost function attained when running the quantum circuit. Our sampling algorithm $\mathcal{A}$ applied to quantum circuit $\mathcal{C}$ produces samples following a distribution that we denote $\mathcal{A}(\mathcal{C})$. The expected cost function of such samples is written as $\mathbb{E}_{z \sim \mathcal{A}(\mathcal{C})}[C(z)]$. As often done in optimization, we will be interested in the ratio of these quantities and our goal will be to identify situations in which it is approximately close to $1$, i.e., the sampling algorithm $\mathcal{A}$ performs essentially as well as the quantum circuit itself. \emph{The quality of the samples obtained is characterized by the ``recovery ratio of the sampling algorithm" $\alpha$, defined as  }

\begin{equation}
    \alpha=\frac{\mathbb{E}_{z\sim \mathcal{A}(\mathcal{C})}[C(z)]}{\mathbb{E}_{z \sim \mathcal{M}([\mathcal{C}])}[C(z)]}
\end{equation}

The quality of the obtained samples will be assessed using the cost function of the optimization problem, rather than Total Variation (TV) distance, as the cost function is more relevant from an optimization perspective. While most prior works focus on simulation quality in terms of TV distance \cite{Aharonov2023, nelson2024}, we emphasize that our approach aligns more closely with the practical objective of finding high-quality solutions to the optimization problem. Nevertheless, it is worth noting that the samples generated by our algorithm retain 2-body marginals that are close to the original distribution of the quantum circuit output.

We also remark that classically obtaining samples from a quantum circuit in the general case is believed to be a hard task \cite{Terhal2004,Bremner2010}. 
Our objective is to investigate how close we can get to the original circuit by sampling from the expectation values and how noise affects our approximation algorithms performance.

\subsection{Sampling using expectation values}

Our main algorithm, based on randomized rounding, comes from semidefinite programming (SDP) relaxations often used in combinatorial optimization, for example in the well-known Goemans-Williamson algorithm \cite{Goemans1995}. Our main algorithm, described formally in Algorithm \ref{samplingAbs}, uses expectation values in the form of the circuit variance-covariance matrix $\Sigma$, such that $\Sigma_{ij}=\tr([\mathcal{C}]Z_iZ_j)$, and the circuit mean vector $\mu$, such that $\mu_i=\tr([\mathcal{C}]Z_i)$. Formal definition of these quantities is given in Def. \ref{defSigmaMu}. The difficulty of computing these quantities in practice is postponed to Section \ref{sectionCov}.

The first step of our algorithm is to sample the multivariate Gaussian distribution with parameters $\mu$ and $\Sigma$. This can be done efficiently \cite{Gentle2009} and produces a sample $y=(y_1,\ldots,y_n) \in \mathbb{R}^n$. Each coordinates of the vector is then rounded to its sign, to produce a vector $z$ such that $z_i=\sign(y_i)$. The obtained sample $z$ is a bitstring in $\{-1,1\}^n$, which we can use as a solution to the original optimization problem. Note that alternative rounding methods, such as rounding based on the amplitudes of the samples \cite{Alon2004}, are also possible.

Because of the simplicity of such algorithms, it is possible to obtain both \textit{a priori} and \textit{a posteriori} bounds on the quality of the samples obtained. By \textit{a priori} bounds we mean those that only make use of properties of the cost function and the quantum circuit (depth, noise, etc.) to obtain analytical guarantees on the performance of the algorithm. By \textit{a posteriori} bounds we mean polynomial-time algorithms that, given a covariance matrix, output a bound on the approximation obtained by running the rounding algorithm on that specific matrix.
Having \textit{a priori} lower bounds on the quality of the samples obtained is particularly important as it allows us to identify regimes in which the rounding performs approximately as well as the (noisy) quantum computer. Thus, for these instances it is preferable to run our classical rounding algorithm than to sample from the quantum computer. Since these bounds are often very problem-specific, we will focus on two widely studied combinatorial optimization problems, namely the Max-Cut problem and the QUBO problem. These problems are formally defined in Sections \ref{sectionMC} and \ref{sectionQUBO}. Our approach is not restricted to 2-body interaction problems, and could be extended to $k$-body problems \cite{Khot2007} or even to problems on qudits, where the nodes can take more values than two values \cite{Frieze1997}. However, in these settings, the quality of the approximation often becomes dependent on the size of the problem, with performance typically degrading as the problem size increases.

We first start with the seminal result of Goemans-Williamson \cite{Goemans1995}, giving us a rigorous already studied bound on how well our randomized rounding Algorithm \ref{samplingAbs} performs on Max-Cut. Let $\mathcal{G}=(V,E)$ be a graph instance with weights $w_{ij}\geq0$ for each edge $(i,j) \in E$. Computing the recovery ratio can be done explicitly and we get that,

\begin{equation}
\begin{aligned}
    \alpha_{MC}=\frac{\mathbb{E}_{z\sim\mathcal{A}(\mathcal{C})}[C(z)]}{\mathbb{E}_{z\sim\mathcal{M}([\mathcal{C}])}[C(z)]}=\frac{\frac{1}{\pi}\sum_{(i,j)\in E}w_{ij}\arccos\Sigma_{ij}}{\frac{1}{2}\sum_{(i,j)\in E}w_{ij}(1-\Sigma_{ij})}&\geq \alpha_{GW}=0.87856\\
\end{aligned}
\end{equation}

Furthermore, if we are able to bound the variance-covariance coefficients such that for all edge $(i,j)\in E$, $\abs{\Sigma_{ij}}\leq \varepsilon\leq x'=0.689$, then we strictly improve on this ratio, and obtain that $\alpha_{MC}>1-f(\varepsilon)$ where $f:x\mapsto (1-\frac{2\arccos x}{\pi(1-x)})\leq 1-\alpha_{GW}$ is such that $f(\varepsilon)=O(\varepsilon)\xrightarrow[\varepsilon \to 0]{} 0$. If we have \textit{a priori} knowledge on the variance-covariance matrix, it is possible to improve this ratio, which goes to 1 as the coefficients get to 0. It is also possible to explicitly compute \textit{a posteriori} how well our rounding performs once the variance-covariance matrix is known. In the noisy setting, we expect the qubits to become less correlated as the noise and depth of the circuit increases. In the extreme noise regime, our sampling algorithm will therefore perform well, which is to be expected since the output of the quantum circuit is random itself. However, we are able to formalise this result to all noise levels using quantum transportation costs inequalities \cite{DePalma2021}, and show that it is possible to characterize the quality of the samples obtained through our rounding algorithm on a noisy circuit. The noise considered is the local depolarizing noise $\mathcal{N}_{DP}$, see Def. \ref{defNoiseDepo}.

\begin{theorem}[Performance of Algorithm \ref{samplingAbs} on noisy Max-Cut] Consider $\mathcal{N}_{DP}$ the local depolarizing channel of strength $p$ acting on our $D$-layer quantum circuit $\mathcal{C}$ on $n$ qubits. Let $\mathcal{G}=(V,E)$ be a regular graph with uniform weights. Denote $[\mathcal{C}]_{\mathcal{N}_{DP}}$ the state prepared by the noisy quantum circuit and $\varepsilon^2 = 2\sqrt{2}(1-p)^D$, such that $\varepsilon\leq x'=0.689$. The samples produced by Algorithm \ref{samplingAbs} are such that,
\begin{equation}
\begin{aligned}
    \alpha_{MC'}=\frac{\mathbb{E}_{z\sim\mathcal{A}(\mathcal{C})}[C(z)]}{\mathbb{E}_{z\sim\mathcal{M}([\mathcal{C}]_{\mathcal{N}_{DP}})}[C(z)]} \geq 1 - h(\varepsilon)
\end{aligned}
\end{equation}
where $h(\varepsilon)\leq1-\alpha_{GW}$ and $h(\varepsilon)=O(\varepsilon)\xrightarrow[\varepsilon \to 0]{} 0$, with $h:x\mapsto 1-x\alpha_{GW}-(1-x)\frac{2\arccos(-x)}{\pi(1+x)}$.
\end{theorem}

This result is obtained by bounding the average variance-covariance coefficient in our graph, using transportation cost tools, detailed in Appendix \ref{appendixTC}, and combining this bound with the analytical expression of the recovery ratio.

It is known that under depolarizing noise, sampling from quantum circuits at $\log n$ depth is possible, as simply sampling from the uniform distribution provides a good approximation \cite{Deshpande2022}. Prior work also demonstrated that quantum advantage is lost at \emph{fixed} depth under noise, with classical algorithms outperforming the quantum ones in such settings \cite{StilckFrana2021}. Our work goes beyond these results by showing that in such regimes, not only is quantum advantage lost, but it is possible to reproduce samples close to those produced by the quantum circuit. Unlike prior approaches, our method does not rely on the noise or depth exceeding a particular threshold, and its performance improves as noise increases. Furthermore, computing the noisy variance-covariance matrix can be done efficiently in average \cite{fontana2023classical, Martinez2025}, offering us a powerful end-to-end method for sampling from quantum circuits, which is presented in Section \ref{sectionCov}.

Another well-known problem we tackle is the QUBO problem, which generalizes Max-Cut by allowing arbitrary signs in the cost function. Similarly to the Max-Cut case, we sample from the quantum circuit using Algorithm \ref{samplingAbs}. Applying the same argument as in~\cite{Nesterov1998}, we obtain a guarantee on the recovery ratio given by the inequality
$ \mathbb{E}_{z\sim \mathcal{A}(\mathcal{C})}[C(z)]/\mathbb{E}_{z\sim \mathcal{M}([\mathcal{C}])}[C(z)]\geq 2/\pi$. Because of the structure of the cost function for QUBO, improving this bound in the noisy case turns out to be more challenging.

\begin{theorem}[Performance of Algorithm \ref{samplingAbs} on noisy QUBO] Let $\mathcal{G}=(V,E)$ be the graph representing our QUBO, with $\Delta$ the maximum degree of $\mathcal{G}$. Consider $\mathcal{N}_{DP}$ the local depolarizing channel of strength $p$ acting on our $D$-layer quantum circuit $\mathcal{C}$ on $n$ qubits.  Denote $[\mathcal{C}]_{\mathcal{N}_{DP}}$ the state prepared by the noisy quantum circuit. The samples produced by Algorithm \ref{samplingAbs} are such that,

\begin{equation}
    \alpha_{QUBO'}=\frac{\mathbb{E}_{z\sim \mathcal{A}(\mathcal{C})}[C(z)]}{ \mathbb{E}_{z\sim \mathcal{M}([\mathcal{C}]_{\mathcal{N}_{DP}})}[C(z)]}=\frac{2\sum_{(i,j)\in E}A_{ij}\arcsin\Sigma_{ij}}{\pi\sum_{(i,j)\in E} A_{ij}\Sigma_{ij}}\geq 1-\frac{\sqrt{2}(2\pi-4)}{\pi}\Delta n (1-p)^D
\end{equation}
\end{theorem}

As mentioned before, it is already well-known that sampling from quantum circuits is possible at $\log n$ depth, and our sampling algorithm seems to require similar depths to improve on the $2/\pi$ original bound. However, it is crucial to highlight that our algorithm offers a consistent method for acquiring samples regardless of the noise conditions, with an improving worst-case guarantee as the depth grows beyond a certain threshold. Furthermore, when simulations are performed, even at a constant depth, the algorithm generates samples that closely match the cost achieved by the quantum circuit. Therefore, utilizing our algorithm proves to be advantageous as it consistently generates accurate samples that closely represent the quantum circuit. Even though sampling from a distribution that is close in TV is hard for certain circuits, our randomized rounding still performs well for any circuit made for combinatorial optimization, shedding light on the limitations of noisy quantum algorithms for such problems.

\section{Notations and background}\label{sec:background}
    \subsection{Pauli matrices}

We denote by $\mathbb{I}, X, Y$ and  $Z$ the usual single-qubit Pauli matrices. We index the qubit $i$ on which this Pauli operator $S$ acts by writing $S_i$. Furthermore, we define a Pauli string $P$ of size $n$ as the tensor product of $n$ Pauli matrices. For simplicity, we omit to write the identity in the tensor product. Therefore, when we write down $Z_iZ_j$, the corresponding Pauli string is $P=\mathbb{I}\otimes \cdots \otimes \underbrace{Z}_{i^{th} qubit} \otimes \cdots \otimes \underbrace{Z}_{j^{th} qubit} \otimes \cdots \otimes \mathbb{I }$. Similarly, $Z_i$ denotes the Pauli string $\mathbb{I}\otimes\cdots\otimes \underbrace{Z}_{i^{th} qubit}\otimes\cdots\mathbb{I}$.
    \subsection{Circuits and observables}
Throughout this work, we aim to obtain samples from a quantum circuit $\mathcal{C}$ designed for combinatorial optimization. This circuit acts on $n$ qubits and is made of $D$ layers of unitary gates acting on one or two qubits. In the form of quantum channels, the circuit is represented as:

\begin{equation}
    \mathcal{C}=\mathcal{C}^{(D)}\circ \mathcal{C}^{(D-1 )}\circ\cdots\circ \mathcal{C}^{(2)}\circ \mathcal{C}^{(1)}
\end{equation}

We will also consider the noisy version of such circuits, where instead of implementing the $i^{th}$ unitary layer $\mathcal{C}^{(i)}$ we implement an altered version $\mathcal{N}^{(i)}\circ \mathcal{C}^{(i)}$ where $\mathcal{N}^{(i)}$ is an unwanted noise channel. \\

\begin{defi}
Let $\mathcal{C}$ be a a noiseless quantum circuit made of $D$ layers of unitary gates. We call the noisy version of $\mathcal{C}$ the circuit $\mathcal{C}'$ affected by the noise channels $\mathcal{N}^{(i)}$ at each layer $i$, such that,

\begin{equation}
    \mathcal{C'}=\mathcal{N}^{(D)}\circ \mathcal{C}^{(D)}\circ \cdots\circ \mathcal{N}^{(2)}\circ\mathcal{C}^{(2)}\circ \mathcal{N}^{(1)}\circ\mathcal{C}^{(1)}
\end{equation}
\end{defi}

In this work, the noise channels will be the same for every layer $i$ and will therefore be noted $\mathcal{N}$ for simplicity. Two noise channels of particular importance will be studied, the depolarizing channel $\mathcal{N}_{DP}$ and the amplitude damping channel $\mathcal{N}_{AD}$, as a combination of these two channels is often considered a good approximation for the noise on current hardware~\cite{Papic2023}.

\begin{defi}\label{defNoiseDepo}
    Let $\rho$ be a quantum state on $d$ qubits. The $2^d$-dimensional depolarizing channel $\mathcal{N}_{DP}^{(d)}$ of strength $p$ is defined as the linear map:

    \begin{equation}
        \mathcal{N}_{DP}^{(d)}(\rho)=(1-p)\rho+p\tr(\rho)\frac{I}{2^d}
    \end{equation}
 
\end{defi}
   Throughout this paper, we will consider local depolarizing noise, corresponding to depolarizing noise acting on each qubit individually. The resulting channel corresponds to the tensor product of these noise channels, which can be denoted $\mathcal{N}_{DP}^{\otimes n}$. Similarly, we define the amplitude damping noise channel as follows:

\begin{defi}\label{defNoiseAD}
    Let $\rho$ be a one qubit quantum state. The amplitude damping channel $\mathcal{N}_{AD}$ of strength $p$ is defined as the linear map:

    \begin{equation}
        \mathcal{N}_{AD}(\rho)=K_0\rho K_0^\dagger+K_1\rho K_1^\dagger
    \end{equation}
    where $K_0=\begin{pmatrix}
1 & 0 \\
0 & \sqrt{1-p} 
\end{pmatrix}$ and $K_1=\begin{pmatrix}
0 & \sqrt{p} \\
0 & 0 
\end{pmatrix}$.
\end{defi}
   
As for the depolarizing noise channel, when dealing with an $n$-qubit system, we will consider the tensor product of $n$ 1-qubit amplitude damping noise channel, denoted $\mathcal{N}_{AD}^{\otimes n}$. When measured, both the output of the noiseless circuit $\mathcal{C}$ and of the noisy circuit $\mathcal{C}'$ represent a bitstring of length $n$ from which we can extract information.

\begin{defi}
    Let $\rho=[\mathcal{C}]$ be the output of a quantum circuit. Let $O$ be an observable of interest. We call the expectation value of the observable $O$, which we denote $\langle O \rangle$, the quantity:

    \begin{equation}
        \langle O \rangle = \Tr(\rho O)
    \end{equation}
\end{defi}

In this paper, the expectation value of the Pauli strings of the form $Z_i$ and $Z_iZ_j$ will be of particular importance. As such, we will introduce the folling notation for them.

\begin{defi}\label{defSigmaMu}
    Let $\rho=[\mathcal{C}]$ be the output density matrix of a potentially noisy quantum circuit. We define the mean vector $\mu\in\R^n$ and the covariance matrix $\Sigma\in\R^{n\times n}$ as:
    \begin{equation}
    \begin{aligned}
     \quad ~\forall i \in \{1,\ldots,n\},& \quad    \mu_i=\Tr(\rho Z_i) \\
      \quad ~\forall (i,j) \in \{1,\ldots,n\}^2, &\quad  \Sigma_{ij}=\Tr(\rho Z_iZ_j)
    \end{aligned}
    \end{equation}
\end{defi} 

In the cases presented in Section \ref{sectionMC} and \ref{sectionQUBO}, the distinction between correlation matrix and variance-covariance matrix will not be made, as the mean vector will be 0 and the standard deviation of each qubit 1. We clearly always have $\Sigma_{ii}=1$ and for $i\neq j$, $\Sigma_{ij}\in[-1,1]$. Explicitly computing these quantities for families of noisy and noiseless circuits is discussed in Section \ref{sectionCov}. The correlation matrix $\Sigma$ thus obtained has a particular structure, as it has only 1s in the diagonal and is positive semidefinite. 

\subsection{Semidefinite programming for combinatorial optimization}\label{sec:SDP_pres}

Semidefinite programming (SDP) is a powerful optimization technique which has been at the center of many state of the art algorithms in various areas of science~\cite{Gartner2012-re}. The SDP relaxation for the problems we will consider, namely Max-Cut and QUBO on $n$ variables, can be written \cite{Goemans1995},

\begin{equation}\label{eq:SDPrelax}
\begin{aligned}
\max_{X} \quad & \tr(C^TX)\\
\textrm{s.t.} \quad & X_{ii}=1, \quad &i&=1,\ldots,n\\
  &X \succeq 0
\end{aligned}
\end{equation}

For Max-Cut, $C$ represents the Laplacian matrix of the weighted graph $\mathcal{G} = (V, E)$ such that,

\begin{equation}\label{eq:MatrixMaxCut}
    C_{ij} =
    \begin{cases} 
         -\frac{1}{4}w_{ij} & \text{if }(i,j)\in E, \\
         \frac{1}{4}\sum_k w_{ik} & \text{if } i=j, \\
        0 & \text{otherwise}.
    \end{cases}
\end{equation}

For the QUBO problem, $C$ represents the cost matrix as defined in Eq. (\ref{eq:QUBOdef}). Solving the SDP relaxation provides an upper bound on the maximum of the cost function and can be done efficiently for problems with tens of thousands of variables. However, the correlation matrix obtained from the SDP solution often does not correspond to feasible solutions in $\{-1,1\}^n$, as there may be no valid distribution in $\{-1,1\}^n$ matching this matrix.
To perform the rounding, the SDP solution is interpreted as a correlation matrix. Sampling is done using a Gaussian distribution $\mathcal{N}(0, \Sigma)$, where $\Sigma$ is the correlation matrix. The samples $x \in \mathbb{R}^n$ are then rounded to their sign:

\begin{equation}
z_i = \text{sign}(x_i), \quad z \in \{-1, 1\}^n,
\end{equation}

yielding valid solutions for the original optimization problem. Alternative rounding techniques can also be employed, such as rounding based on the amplitudes of the samples rather than their sign \cite{Alon2004}.

The final step involves analyzing the performance of the rounding method. The quality of approximate solutions obtained via randomized rounding is typically evaluated using the approximation ratio, which compares the expected solution value to the optimal one. Rigorous performance analysis is crucial for understanding the trade-offs and guarantees offered by the rounding algorithm. Previous studies have analyzed this process for Max-Cut and QUBO, demonstrating that the described pipeline produces samples with an approximation ratio of at least $0.878$ for Max-Cut \cite{Goemans1995} and $2/\pi$ for QUBO \cite{Nesterov1998}.

For both problems, applying the rounding algorithm to the SDP relaxation of Eq. (\ref{eq:SDPrelax}) instead of applying it to the circuit variance-covariance matrix of Def. \ref{defSigmaMu} seems to be a better alternative. Indeed, by definition, the correlation matrix obtained by solving the SDP relaxation gives an expected cost that is greater than the quantum circuit—even though it remains uncertain which solution yields better outcomes after the rounding has been applied. However, our sampling method concentrates on creating samples that accurately represent the quantum circuit, and we utilize the cost function to evaluate this accuracy. By using the variance-covariance matrix produced by the quantum circuit, we ensure that the samples produced follow a distribution with marginals close to the quantum circuit output. In this regard, our sampling strategy is distinct from other studies that merely attempt to outperform the quantum algorithm \cite{StilckFrana2021}. 

In the following sections, we will adapt the randomized rounding techniques presented to sample from quantum circuits.

\section{Sampling from Max-Cut circuits}\label{sectionMC}

\subsection{The Max-Cut problem}

The Max-Cut problem is a fundamental combinatorial optimization problem, defined as follows. Consider a weighted graph instance $\mathcal{G} = (V, E)$, where $V$ is the set of vertices and $E$ is the set of edges. The objective is to partition the vertex set into two subsets (commonly labeled as $+$ and $-$) to maximize the following objective function:

\begin{equation}\label{eq6}
\begin{aligned}
\max_{z} \quad C(z) = & \frac{1}{2}\sum_{(i,j)\in E}w_{ij}(1-z_iz_j)\\
\textrm{s.t.} \quad & z_i \in \{-1,1\} \quad \forall i \in V\\
\end{aligned}
\end{equation}

where the edge weights $w_{ij} > 0$ for $(i, j) \in E$. The more general case, where edge weights can have arbitrary signs \cite{Charikar}, corresponds to QUBO problems, which are studied in the next section. The Max-Cut problem has several applications across domains such as statistical physics, machine learning, and various other fields \cite{maxCutML, maxCutPhysics}. It is well-known that solving this problem is NP-hard \cite{Khot2004}, and numerous heuristics have been developed to find good approximate solutions.

One of the most significant contributions to solving the Max-Cut problem is the Goemans-Williamson algorithm \cite{Goemans1995}. This randomized approximation algorithm utilizes the SDP relaxation of the original problem combined with randomized rounding techniques introduced in Section \ref{sec:SDP_pres}. The algorithm guarantees an expected value of at least 0.87856 times the optimal value. Furthermore, it has been shown that, assuming the Unique Games Conjecture, it is NP-hard to achieve a better approximation ratio for the general case \cite{Khot2004}.

Given the importance of combinatorial optimization problems, many proposals were made to try and obtain better solutions to the Max-Cut problem on quantum computers \cite{farhi2014quantum,Wurtz2021}. This problem can indeed easily be mapped to a quantum computer, by transforming it into an Ising Hamiltonian. Maximizing the Max-Cut cost function is equivalent to minimizing the Ising Hamiltonian $H$.

\begin{equation}
    H=\frac{1}{2}\sum_{i<j}w_{ij}Z_iZ_j
\end{equation}

\noindent Minimising the energy associated to this Hamiltonian has been done using QAOA and guarantees were obtained even at low-depth for certain classes of graph \cite{Harrigan2021,farhi2014quantum,Wurtz2021}. Obtaining solutions to the associated combinatorial optimization problem requires running the quantum circuit and measuring all qubits which might expose the solution obtained to the inherent noise of the quantum device, deteriorating it. Although various techniques have been developed to recover noiseless expectation values of observables, such as classical lightcone simulations \cite{Evenbly2007} or error mitigation methods \cite{Cai2023}, and noisy expectation values, through classical simulations \cite{fontana2023classical,Martinez2025}, retrieving noiseless or noisy samples directly remains a challenging task \cite{arunachalam2023role}.
The distinction between observables and samples is critical: observables represent properties of the solution, such as the cut achieved by it, whereas samples correspond to the solution itself, \textit{i.e.}, the assignment of nodes. In optimization, obtaining the solution is of paramount importance, as it directly corresponds to the practical implementation required for real-world applications.

In the following sections, we will study how well we can sample from these optimization-designed quantum circuits by only using expectation values. We present a method that enables the recovery of samples from the quantum circuit and provide both analytical guarantees and numerical evidence demonstrating that these samples are faithful—meaning their cost function closely matches that of the original samples produced by the quantum circuit. More importantly, we show that if the circuit is subject to noise, then this method allows to closely sample from the quantum circuit. Sampling from noisy circuits is particularly important, as prior work has shown that repeatedly sampling and selecting the best samples can recover the noiseless expectation values with high probability \cite{Barron2024}. Our numerical experiments in Section \ref{sectionNumerics} further demonstrate that our sampling method achieves the same recovery of expectation values without requiring direct access to the noisy quantum circuit, providing a reliable and efficient alternative to sampling the quantum circuit.

\subsection{Sampling Max-Cut circuits in the noiseless setting}

Consider a quantum circuit $\mathcal{C}$ tackling the Max-Cut problem. We suppose in the following that we are able to compute the circuit variance-covariance matrix $\Sigma$ and the circuit mean vector $\mu$, as defined in Def. \ref{defSigmaMu}. The difficulty of computing this quantity is discussed later on in Section \ref{sectionCov}. To sample from the quantum circuit, we first introduce explicitly the sampling algorithm we will be using.

\begin{alg}[Absolute randomized rounding]\label{samplingAbs}
\hfill \break
\textit{Input:} Quantum circuit $\mathcal{C}$ \\
\textit{Output:} Samples close to $\mathcal{C}$
\begin{enumerate}
    \item Compute the variance-covariance matrix $\Sigma_{ij}=\Tr([\mathcal{C}]Z_iZ_j)$
    \item Sample from the multivariate Gaussian distribution $(y_1,...,y_n) \sim \mathcal{N}\big(0,\Sigma\big)$ 
    \item Assign to each node $i$ the value $z_i= \left\{
    \begin{array}{ll}
        +1 & \mbox{if } y_i\geq0 \\
        -1 & \mbox{if } y_i<0
    \end{array}
\right.$
    \item Return $z = (z_1, \ldots, z_n)$
\end{enumerate}
\end{alg}

The intuition behind why such sampling algorithm would work well is that the random vector generated from $\mathcal{N}\big(0,\Sigma\big)$ is likely to preserve the relationships between the vertices as described by $\Sigma$. In other words, vertices that are more positively correlated (i.e., have larger values in $\Sigma$) are more likely to have the same sign in the sampled vector, and vertices that are negatively correlated (i.e., have smaller values in $\Sigma$) are more likely to have different signs. 
By rounding the obtained result to the sign, we are essentially exploiting the underlying structure of the graph as captured by the variance-covariance matrix $\Sigma$. This allows us to find a partition that approximates the quantum circuit solution to Max-Cut.

The performance of Algorithm \ref{samplingAbs} on the Max-Cut problem has been explicitly studied in prior work. In fact, Algorithm \ref{samplingAbs} is essentially the Goemans-Williamson algorithm \cite{Goemans1995} for Max-Cut presented in Section \ref{sec:SDP_pres}, with the key difference being that instead of using the covariance matrix derived from the solution of an appropriate SDP relaxation, the covariance matrix of the quantum circuit itself is utilized.
We briefly outline the proof techniques here, as they will be essential for analyzing the performance of our sampling algorithm in the presence of noise. Let $\rho = [\mathcal{C}]$ represent the prepared quantum state, and let $z \sim \mathcal{M}([\mathcal{C}])$ denote the bitstring probability distribution obtained by measuring $\rho$. For the purposes of this analysis, we assume that the mean of each qubit is $0$, since the cut depends solely on the covariance of the qubits. The average cut $\mathbb{E}_{z \sim \mathcal{M}([\mathcal{C}])}[C(z)]$ achieved by this circuit can then be computed as follows:

\begin{equation}
\begin{aligned}
    \mathbb{E}_{z\sim\mathcal{M}([\mathcal{C}])}[C(z)]&=\mathbb{E}\big[\frac{1}{2}\sum_{(i,j)\in E}w_{ij}(1-z_iz_j)\big]\\
    &=\frac{1}{2}\sum_{(i,j)\in E}w_{ij}(1- \mathbb{E}[z_iz_j])\\
    &=\frac{1}{2}\sum_{(i,j)\in E}w_{ij}(1-\Sigma_{ij})
\end{aligned}
\end{equation}
    
where the last equality holds since the circuit mean vector is 0. To compute the cut achieved by our rounding algorithm, the following well-known lemma \cite{cramer1996} is necessary:

\begin{lem}\label{lemmaQuadrant}
    Let $(Y_i,Y_j)$ have a bivariate normal distribution with correlation $\Sigma_{ij}$. Then:

    \begin{equation}
        \mathbb{P}(Y_i>0, Y_j>0)=\frac{1}{4}+\frac{1}{2\pi}\arcsin\Sigma_{ij}
    \end{equation}
\end{lem}

In the non-centered case, while this quantity can be computed numerically through various methods \cite{Miwa2003}, no analytical results are currently known. As a result, the performance of the rounding algorithm in such biased cases is typically studied through numerical simulations. Using this lemma, it is direct to determine the cut achieved in average by our rounding algorithm.

Let $(z_1,\ldots,z_n) \sim \mathcal{A}(\mathcal{C})$ be the output of Algorithm \ref{samplingAbs}. The cut can be computed explicitly by considering each edge individually and determining the probability that it is in the cut \cite{Goemans1995}:

\begin{equation}
\begin{aligned}
    \mathbb{E}_{z\sim\mathcal{A}(\mathcal{C})}[C(z)]=\sum_{(i,j)\in E}w_{ij}\frac{1}{\pi}\arccos\Sigma_{ij}
    \end{aligned}
\end{equation}

A natural question that arises, is how close this cut is to the one of our original circuit? To answer this question, we can study the recovery ratio of our sampling algorithm defined by $\alpha_{MC}=\mathbb{E}_{z\sim\mathcal{A}(\mathcal{C})}[C(z)]/\mathbb{E}_{z\sim\mathcal{M}([\mathcal{C}])}[C(z)]$. This was also done by Goemans-Williamson \cite{Goemans1995}, which showed that the performance of the rounding algorithm is at least $0.87856$.

\begin{lem} [Performance of Algorithm \ref{samplingAbs} on Max-Cut] \label{propMCnoiseless}Consider a weighted graph $\mathcal{G}=(V,E)$ where $w_{ij}\geq 0$ denotes the weight of edge $(i,j)$. Denote $\Sigma_{ij}=\Tr([\mathcal{C}]Z_iZ_j)$ the correlation matrix of the state prepared by quantum circuit $\mathcal{C}$. The Goemans-Williamson \cite{Goemans1995} result states,

\begin{equation}\label{eq:maxcut_worstcase}
\begin{aligned}
\alpha_{MC}=\frac{\mathbb{E}_{z\sim\mathcal{A}(\mathcal{C})}[C(z)]}{\mathbb{E}_{z\sim\mathcal{M}([\mathcal{C}])}[C(z)]}=\frac{\frac{1}{\pi}\sum_{(i,j)\in E}w_{ij}\arccos\Sigma_{ij}}{\frac{1}{2}\sum_{(i,j)\in E}w_{ij}(1-\Sigma_{ij})}&\geq \alpha_{GW}\approx 0.87856\\
\end{aligned}
\end{equation}

\noindent Suppose additionally that for all edges $(i,j)$, $\abs{\Sigma_{ij}}\leq\varepsilon\leq  0.689$. Then the result can be adapted to show that,

\begin{equation}
\begin{aligned}
    \alpha_{MC}=\frac{\mathbb{E}_{z\sim\mathcal{A}(\mathcal{C})}[C(z)]}{\mathbb{E}_{z\sim\mathcal{M}([\mathcal{C}])}[C(z)]}=\frac{\frac{1}{\pi}\sum_{(i,j)\in E}w_{ij}\arccos\Sigma_{ij}}{\frac{1}{2}\sum_{(i,j)\in E}w_{ij}(1-\Sigma_{ij})}& \geq 1-f(\varepsilon)
\end{aligned}
\end{equation}

\noindent where $f:x\mapsto (1-\frac{2\arccos x}{\pi(1-x)})\leq 1-\alpha_{GW}$ is such that $f(\varepsilon)=O(\varepsilon)\xrightarrow[\varepsilon \to 0]{} 0$
\end{lem}

\begin{proof}
As we will use the proof further on, we present it here. The key aspect is studying the function $g: x\mapsto \frac{2\arccos x}{\pi(1-x)}$. A quick study of the variations of $g$ shows that the ratio decreases until it reaches the $0.87856$ value at $x'/\frac{(x'+1)\arccos(x')-\sqrt{1-x'^2}}{x'^2-1}=0$, which corresponds to $x'\approx -0.689$. After this critical point, the ratio $g$ is strictly increasing. Figure \ref{plotFunctiong} shows the plot of the function $g$, bounded from below by $\alpha_{GW}$.

\begin{figure}[h]
  \centering
  \includegraphics[width=0.75\textwidth]{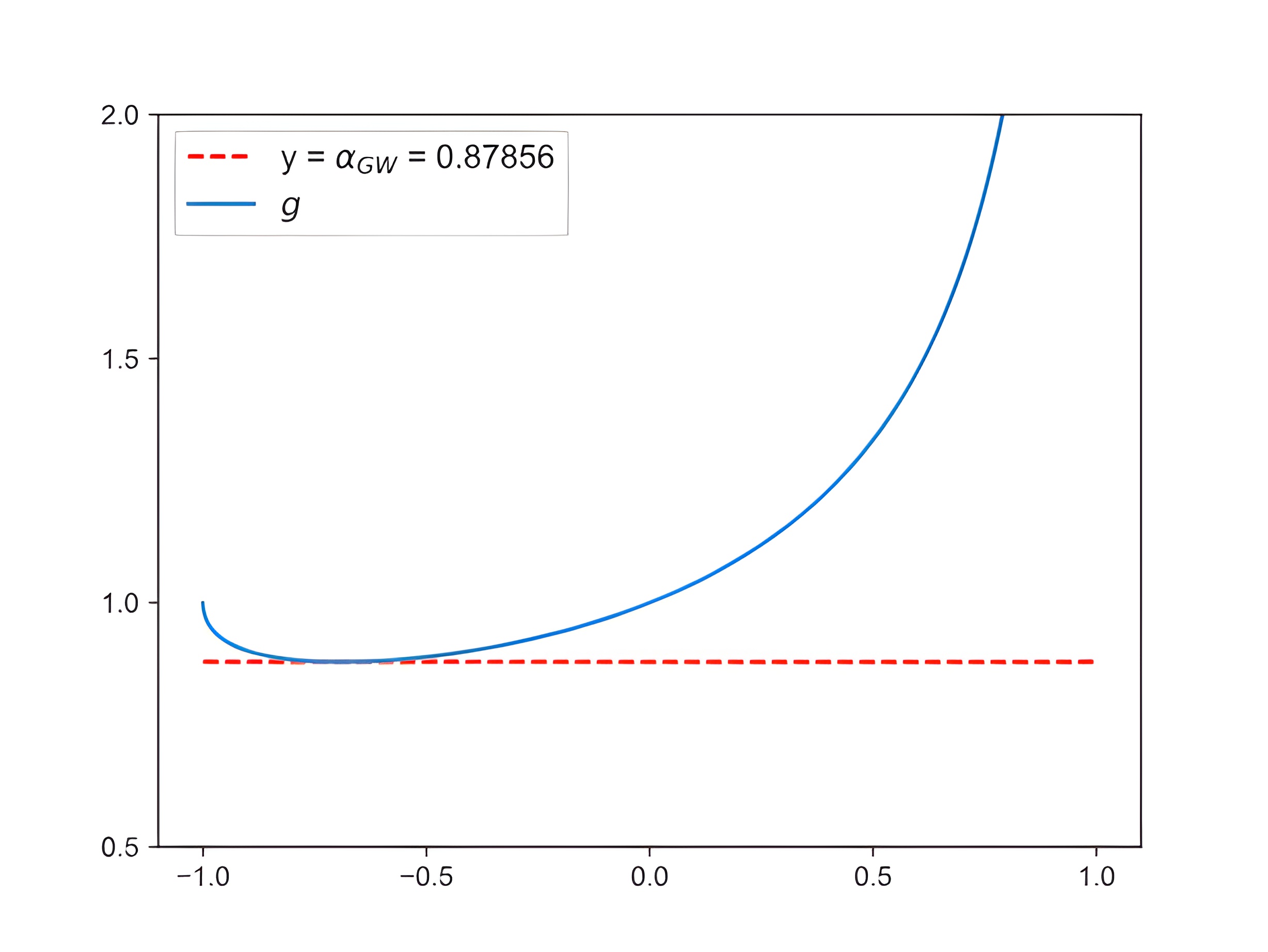}
  \caption{Value of the ratio $g$, lower bounded on $[-1,1]$ by 0.87856. The lower bound on the recovery ratio of Algorithm \ref{samplingAbs} is reached by $\Sigma_{ij}\approx-0.689$ for all $(i,j)\in E$.}\label{plotFunctiong}
\end{figure}

\noindent This result allows us to analytically obtain the theorem above from the ``worst-case scenario", and we get that,

\begin{equation}
\alpha_{MC}=\frac{\mathbb{E}_{z\sim\mathcal{A}(\mathcal{C})}[C(z)]}{\mathbb{E}_{z\sim\mathcal{M}([\mathcal{C}])}[C(z)]}=\frac{\frac{1}{\pi}\sum_{(i,j)\in E}w_{ij}\arccos\Sigma_{ij}}{\frac{1}{2}\sum_{(i,j)\in E}w_{ij}(1-\Sigma_{ij})}\geq \alpha_{GW}=0.87856
\end{equation}

\end{proof}

The recovery ratio, which can be computed explicitly given the correlation matrix, will almost always be better than this ratio, as seen in Section \ref{sectionNumerics}. Furthermore, if we are able to bound the coefficients of $\Sigma$, then this ratio improves and increases to 1 as the correlations tend to 0, as expected. 

\subsection{Sampling Max-Cut circuits in the noisy setting}\label{sectionMCnoise}

Let us now examine the case where the quantum circuit $\mathcal{C}$ is affected by noise. Since the performance of our algorithm is entirely determined by the values of the variance-covariance matrix, the impact of noise can be analyzed by studying its effects on this matrix. As the circuit depth and/or the noise strength increase, the behavior of the variance-covariance matrix can be inferred based on the characteristics of the noise channel.

For local depolarizing noise, we expect the outputs to become uncorrelated as the depth and noise strength increase. In the asymptotic regime, this will result in a correlation matrix $\Sigma = I$, where $I$ is the identity matrix.

In the case of amplitude damping, as the noise strength increases, we expect all output qubits to converge to the $\ket{0}$ state. Denoting $J$ as the matrix filled with ones, the correlation matrix will converge to $\Sigma = \mathbf{1} \cdot \mathbf{1}^T = J$. In Section \ref{sectionNumerics}, we will provide numerical evidence to support this behavior for amplitude damping even in the small noise regime, and here analytically study the performance of our randomized algorithm under depolarizing noise.

When the circuit is subjected solely to depolarizing noise, we can obtain results even at fixed depth, leveraging recent advances in transportation cost inequalities \cite{Palma2023}. Indeed, it is possible to bound how close the state obtained from our noisy circuit is to the maximally mixed state in term of energy, and we prove in the Appendix \ref{appendixTC} the following theorem:

\begin{theorem}[Impact of depolarizing noise on the covariance matrix]\label{theoremBoundCov}
    Consider $\mathcal{N}_{DP}$ the local depolarizing channel of strength $p$ and $\rho=[\mathcal{C}]$ the state prepared by the quantum circuit. Denote $\Delta$ the maximum degree of the graph $\mathcal{G}=(V,E)$. The resulting noisy circuit produces variance-covariance matrix coefficients $\Sigma_{ij}=\tr([\mathcal{C}]_{\mathcal{N}_{DP}}Z_iZ_j)$ such that:
\begin{equation}
    \frac{1}{\Delta n}\sum_{(i,j)\in E}|\Sigma_{ij}| \leq \sqrt{2}(1-p)^D
\end{equation}

\noindent where $D$ corresponds to the depth of our quantum circuit.
\end{theorem}

In the case where our graph is $\Delta$-regular, this result is easily seen to be a bound on the average value of the absolute variance-covariance coefficients value. We can obtain a similar bound in the case where the graph isn't regular, by noting that this theorem implies that 

\begin{equation}
    \frac{1}{|E|}\sum_{(i,j)\in E}|\Sigma_{ij}|\leq 2\sqrt{2}\Delta(1-p)^D
\end{equation}

Indeed, in a fully connected graph, it is clear that $|E|\geq n-1$. This allows us to derive a more general bound on the average correlation coefficients.

\begin{equation}
\begin{aligned}
   \frac{1}{|E|}\sum_{(i,j)\in E}|\Sigma_{ij}|&\leq \frac{1}{n-1}\sum_{(i,j)\in E}|\Sigma_{ij}|\leq \sqrt{2} \Delta(1+\frac{1}{n-1})(1-p)^D\\
   &\leq 2\sqrt{2}\Delta(1-p)^D
\end{aligned}
\end{equation}

For simplicity, we consider the case where the weights of the graph are all equal. Since we have a bound on the average correlation coefficient, we can use Markov inequality to divide the correlation coefficients in two sets. By doing this, we show that in the presence of noise, our randomized rounding algorithm beats the Goemans-Williamson constant at \textit{fixed depth}, producing samples of energy abitrarily close to the original quantum circuit in depth indepent of the graph size. This result is sumarised in the following theorem. 

\begin{theorem}[Performance of Algorithm \ref{samplingAbs} on noisy Max-Cut]\label{theoremNoisyMaxcut} Let $\mathcal{G}=(V,E)$ be a regular graph with uniform weights. Denote $[\mathcal{C}]_{\mathcal{N}_{DP}}$ the state prepared by our noisy quantum circuit with noise strength $p$. Denote $\varepsilon^2 = 2\sqrt{2}(1-p)^D$ the bound achieved by the average correlation coefficient. Recall that $g$ is the function $g: x\mapsto \frac{2\arccos x}{\pi(1-x)}$ introduced in Figure \ref{plotFunctiong}, such that $g(0)=1$. If $\varepsilon\leq x'$, then

\begin{equation}
\begin{aligned}
    \alpha_{MC'}=\frac{\mathbb{E}_{z\sim\mathcal{A}(\mathcal{C})}[C(z)]}{\mathbb{E}_{z\sim\mathcal{M}([\mathcal{C}]_{\mathcal{N}_{DP}})}[C(z)]}&\geq(1-\varepsilon)g(-\varepsilon)+ \varepsilon\alpha_{GW}\\
    & \geq 1 - h(\varepsilon)
\end{aligned}
\end{equation}
where $h(\varepsilon)\leq1-\alpha_{GW}$ and $h(\varepsilon)=O(\varepsilon)\xrightarrow[\varepsilon \to 0]{} 0$, with $h:x\mapsto 1-x\alpha_{GW}-(1-x)\frac{2\arccos(-x)}{\pi(1+x)}$. Note that in the case where $\mathcal{G}$ isn't regular, $\varepsilon^2$ differs by a factor $\Delta$, where $\Delta$ is the maximum degree of $\mathcal{G}$.
\end{theorem}

\begin{proof}

For a given $c\in]0,1[$, the proportion of correlation coefficients greater than $\varepsilon^2/c$ is at most $c$. We can therefore divide the correlation coefficients in two groups and bound the cut achieved:

\begin{equation}
    \sum_{(i,j)\in E}\frac{1}{\pi}\arccos\Sigma_{ij}\geq \alpha_{GW}\sum_{(i,j)/|\Sigma_{ij}|\geq \varepsilon^2/c}\frac{1}{2}(1-\Sigma_{ij})+g(-\varepsilon^2/c)\sum_{(i,j)/|\Sigma_{ij}|< \varepsilon^2/c}\frac{1}{2}(1-\Sigma_{ij})
\end{equation}

\noindent This equation is true as long as we have $\varepsilon^2/c\leq x'$. Since all the terms in the sums are smaller or equal to 1, we have that,

\begin{equation}\label{lastEqMCnoisy}
\begin{aligned}
    \sum_{(i,j)\in E}\frac{1}{\pi}\arccos\Sigma_{ij}&\geq \alpha_{GW}\frac{\sum_{(i,j)/|\Sigma_{ij}|\geq \varepsilon^2/c}\frac{1}{2}(1-\Sigma_{ij})}{\sum_{(i,j)\in E}\frac{1}{2}(1-\Sigma_{ij})}\sum_{(i,j)\in E}\frac{1}{2}(1-\Sigma_{ij})\\&+g(-\varepsilon^2/c)\frac{\sum_{(i,j)/|\Sigma_{ij}|<\varepsilon^2/c}\frac{1}{2}(1-\Sigma_{ij})}{\sum_{(i,j)\in E}\frac{1}{2}(1-\Sigma_{ij})}\sum_{(i,j)\in E}\frac{1}{2}(1-\Sigma_{ij})\\
    &\geq c \alpha_{GW} \sum_{(i,j)\in E}\frac{1}{2}(1-\Sigma_{ij})+(1-c)g(-\varepsilon^2/c)\sum_{(i,j)\in E}\frac{1}{2}(1-\Sigma_{ij})
\end{aligned}
\end{equation}

\noindent Eq. (\ref{lastEqMCnoisy}) also holds for $c=\varepsilon$ as long as $\varepsilon\leq x'$. Substituing gives us the final result:

\begin{equation}
    \frac{\mathbb{E}_{z\sim\mathcal{A}}[C(z)]}{\mathbb{E}_{z\sim\mathcal{M}(\rho)}[C(z)]}\geq \varepsilon\alpha_{GW}+(1-\varepsilon)g(-\varepsilon)
\end{equation}
\end{proof}

Note that the recovery ratio goes to one as the depth increases, and the well-known Goemans-Williamson result is bettered even at \emph{fixed depth}, as seen in Figure \ref{fig:MCdepth}.

\begin{figure}[h]
  \centering
  \includegraphics[width=0.75\textwidth]{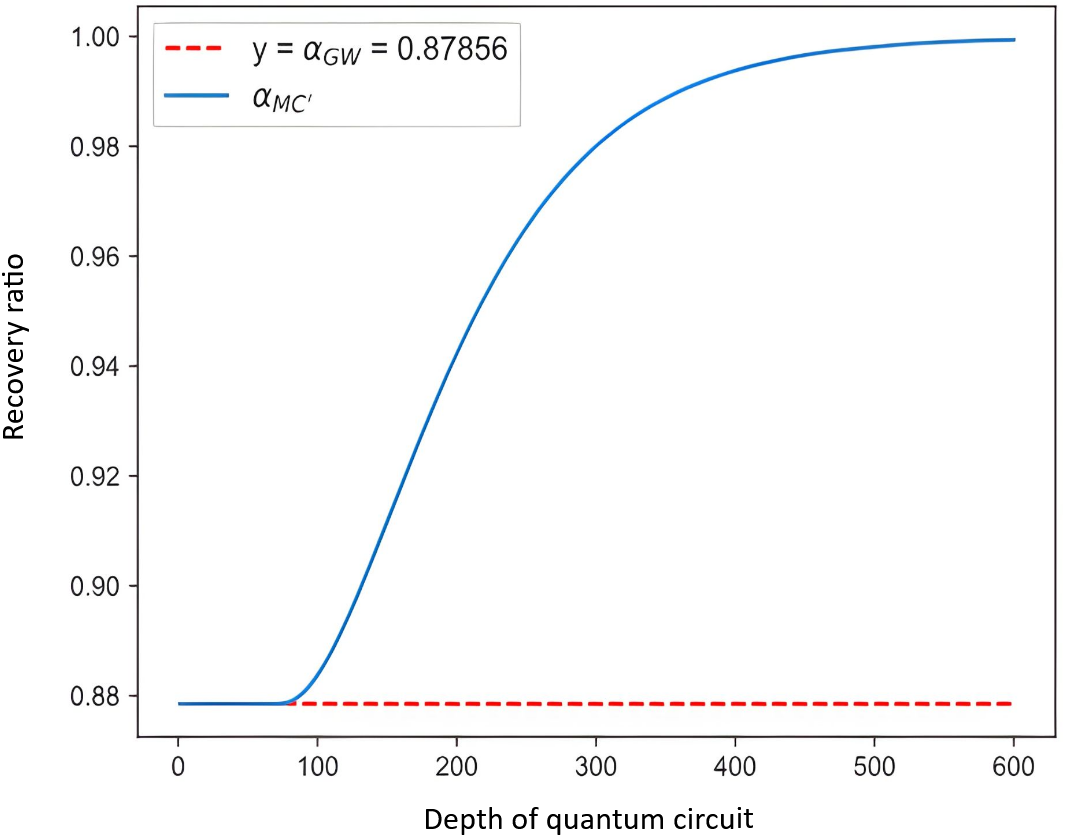}
  \caption{\textit{A priori} recovery ratio for Max-Cut in the noisy regime, as obtained in Th. \ref{theoremNoisyMaxcut}. The depolarizing noise is applied to all qubits in each layer and the noise strength is $p=0.024$.}\label{fig:MCdepth}
\end{figure}

Previously, it was established that sampling from quantum circuits under depolarizing noise is achievable at $\log n$ depth \cite{Deshpande2022}, as simply sampling from the uniform distribution provides a suitable approximation. It was also demonstrated that quantum advantage disappears at fixed depth under such noise channels since classical algorithms outperform their quantum counterparts in this context \cite{StilckFrana2021}. Our work significantly enhances these findings, as we can obtain samples that are remarkably close to those generated by the quantum circuit at fixed depth with a simple algorithm. This not only results in the loss of quantum advantage in this domain, but also enables the replication of the quantum circuit's output. Furthermore, recent work has shown that computing noisy expectation values can be done classically for \emph{most} quantum circuits under almost any single-qubit noise channel \cite{Martinez2025}. Combined with our sampling algorithm, this allows for a robust end-to-end approach for sampling from noisy quantum circuits.

In addition, it should be noted that all our results can easily be adapted to the setting of quantum annealers subjected to local, unital noise by considering continuous time noisy dynamics instead of discrete quantum circuits.

\section{Sampling from QUBO circuits}\label{sectionQUBO}

\subsection{The QUBO problem}

We present another well-known combinatorial optimization problem, with broader applications, the Quadratic Unconstrained Binary Optimization Problem (QUBO). QUBO is a typical unconstrained discrete optimization problem given as follows:

\begin{equation}
    \max_{x\in\{0,1\}^n} x^TQx
\end{equation}

\noindent where $Q \in \mathbb{R}^{n \times n}$ is the cost matrix for $n$ variables. Many optimization problems can be reformulated as QUBOs by adding constraints as quadratic penalty terms or converting higher-order terms into additional binary variables \cite{Lucas2014}.

QUBOs are well-suited for quantum computers as binary variables can be replaced with spin variables, transforming the problem into a ground-state search. In physics and chemistry, these formulations are known as Ising spin glasses \cite{Nishimori2001}. QUBOs are NP-hard, and numerous quantum algorithms were designed to tackle such problems \cite{Peruzzo2014,farhi2014quantum}.

Although Max-Cut can be expressed as a QUBO problem, analyzing the recovery ratio of our sampling approach requires a distinct treatment. Specifically, given a matrix $A \in \mathbb{R}^{n \times n}$, we consider the optimization problem:

\begin{equation}\label{eq:QUBOdef}
    \max_{z\in\{-1,1\}^n}z^TAz=\sum_{i=1}^n\sum_{j=1}^nz_iA_{ij}z_j
\end{equation}

 The optimization problem is often depicted in the form of a graph $\mathcal{G}=(V,E)$, where the interaction between two nodes $i$ and $j$ is represented by the term $A_{ij}$. In the following sections, we analyze sampling from QUBO circuits under both noisy and noiseless regimes.

\subsection{Sampling QUBO circuits in the noiseless setting}

As previously noted, while QUBO shares similarities with Max-Cut, analyzing the recovery ratio requires a distinct approach that leverages the specific structure of the cost function. As a result, analyzing QUBO requires distinct tools and a tailored analytical approach.
Assume we have access to the two-body marginals of a quantum state $\rho$ prepared by a quantum circuit $\mathcal{C}$. The QUBO cost function achieved by the quantum circuit can be directly computed as follows,

\begin{equation}
    \mathbb{E}_{z\sim \mathcal{M}(\rho)}[C(z)]=\sum_{(i,j)\in E} A_{ij}\Sigma_{ij}=\tr(A\Sigma)
\end{equation}

In this section, we consider the case where the cost matrix $A$ is positive definite. In this scenario, the expected values of both the rounding procedure and the original optimization problem are guaranteed to be non-negative, as the covariance matrix $\Sigma$ is also positive definite. It is worth noting that $A$ can always be shifted by adding a multiple of the identity matrix to make it positive definite without altering the optimal solution of the problem. However, this adjustment will result in a different recovery ratio being achieved.

The performance analysis of the randomized rounding Algorithm \ref{samplingAbs} for QUBO has been previously conducted by \cite{Nesterov1998}. We provide a brief overview of this analysis here, as it will be valuable for understanding the algorithm's behavior under noise. Using Lemma \ref{lemmaQuadrant}, the cost function achieved by the randomized rounding Algorithm \ref{samplingAbs} can be explicitly computed.,

\begin{equation}
    \mathbb{E}_{z\sim \mathcal{A}}[C(z)]=\frac{2}{\pi}\sum_{(i,j)\in E}A_{ij}\arcsin\Sigma_{ij}
\end{equation}

Since we have set $A$ to be positive definite, the ratio of both expected values is well-defined. We wish to find the greatest value $\alpha\geq0$ such that

\begin{equation}\label{eq:condition_approx_qubo}
    \frac{2}{\pi}\sum_{(i,j)\in E}A_{ij}\arcsin\Sigma_{ij}\geq \alpha  \sum_{(i,j)\in E} A_{ij}\Sigma_{ij}
\end{equation}

Since we suppose $A$ to be positive semidefinite, we only need to study under which conditions on $\alpha$ the matrix $M_\alpha$ s.t. $M^{ij}_\alpha=\frac{2}{\pi}\arcsin\Sigma_{ij}-\alpha\Sigma_{ij}$ is in $S_n^+$, where $S_n^+$ denotes the cone of positive semidefinite matrices.

This analysis has been done before by \cite{Nesterov1998} and using the Taylor series expansion of $\arcsin x$ when $x\in[-1,1]$ and Schur product theorem, it was shown that  $M_{2/\pi} \in S_n^+$, and that,

\begin{equation}\label{eq:QUBO_worstcase}
    \frac{\mathbb{E}_{z\sim \mathcal{A}}[C(z)]}{ \mathbb{E}_{z\sim \mathcal{M}(\rho)}[C(z)]}\geq \frac{2}{\pi}
\end{equation}

In practice on random instances, the recovery ratio reached is often much better than this lower bound, as seen in Section \ref{sectionNumerics}.

\subsection{Sampling QUBO circuits in the noisy setting}

Eq. (\ref{eq:condition_approx_qubo}) tells us that analyzing the recovery ratio of our sampling algorithm on QUBO is equivalent to determining the conditions under which the following matrix, $M_\alpha$, is positive semidefinite.

\begin{equation}
  M_\alpha =
  \left[ {\begin{array}{cccc}
    1-\alpha & \frac{2}{\pi}\arcsin\Sigma_{12}-\alpha\Sigma_{12} & \cdots & \frac{2}{\pi}\arcsin\Sigma_{1n}-\alpha\Sigma_{1n}\\
    \frac{2}{\pi}\arcsin\Sigma_{21}-\alpha\Sigma_{21} & 1-\alpha & \cdots & \frac{2}{\pi}\arcsin\Sigma_{2n}-\alpha\Sigma_{2n}\\
    \vdots & \vdots & \ddots & \vdots\\
    \frac{2}{\pi}\arcsin\Sigma_{n1}-\alpha\Sigma_{n1} & \frac{2}{\pi}\arcsin\Sigma_{n2}-\alpha\Sigma_{n2} & \cdots & 1-\alpha\\
  \end{array} } \right]
\end{equation}

A first remark can be be made using Gerschgorin circle theorem \cite{gerschgorin31}. Let $R_i$ be the sum of the absolute values of the non-diagonal entries of each row, \textit{i.e.}

\begin{equation}
R_i=\sum_{j\neq i} |\frac{2}{\pi}\arcsin\Sigma_{ij}-\alpha\Sigma_{ij}|
\end{equation}

Every eigenvalue of $M_\alpha$ lies within one of the discs $D(1-\alpha,R_i)$. A sufficient condition for $M_\alpha$ to be positive semidefinite is therefore that for all $R_i<1-\alpha$. Using this fact and Th. \ref{theoremBoundCov}, one can show the following result.

\begin{theorem}[Performance of Algorithm \ref{samplingAbs} on noisy QUBO]\label{theoremNoisyQUBO} Let $\rho=[\mathcal{C}]_{\mathcal{N}_{DP}}$ be a quantum state on $n$ qubits prepared by a noisy circuit of depth $D$ with noise strength $p$. Using randomized rounding Algorithm \ref{samplingAbs}, the recovery ratio of our sampling scheme is at least:

\begin{equation}
    \frac{\mathbb{E}_{z\sim \mathcal{A}}[C(z)]}{ \mathbb{E}_{z\sim \mathcal{M}(\rho)}[C(z)]}=\frac{2\sum_{(i,j)\in E}A_{ij}\arcsin\Sigma_{ij}}{\pi\sum_{(i,j)\in E} A_{ij}\Sigma_{ij}}\geq 1-\frac{\sqrt{2}(2\pi-4)}{\pi}\Delta n (1-p)^D
\end{equation}
\end{theorem}

\begin{proof}
    We first start by bounding the radius of the circles of Gershgorin theorem. Using known upper bounds on the rest of the Taylor series expansion of $\arcsin$ \cite{Bagul2021}, we upper bound the radius,
\begin{equation}
\begin{aligned}
    R_i&=\sum_{j\neq i} |\frac{2}{\pi}\arcsin\Sigma_{ij}-\alpha\Sigma_{ij}|\\
    &\leq (\alpha-\frac{2}{\pi})\sum_{j\neq i} |\Sigma_{ij}| + \frac{\pi-2}{\pi} \sum_{j\neq i}|\Sigma_{ij}^3|\\
    &\leq (\alpha-\frac{2}{\pi}+ \frac{\pi-2}{\pi})\sum_{j\neq i} |\Sigma_{ij}|\\
    &\leq (\alpha-\frac{4-\pi}{\pi})\sum_{j\neq i} |\Sigma_{ij}|
\end{aligned}
\end{equation}

Th.~\ref{theoremBoundCov} can be used to bound the sum of the absolute variance-covariance coefficients and obtain that

\begin{equation}
    R_i\leq \sqrt{2}(\alpha-\frac{4-\pi}{\pi})\Delta n (1-p)^D
\end{equation}

As long as $R_i \leq 1-\alpha$, we obtain an $\alpha$-approximation algorithm. We can rewrite this condition,

\begin{equation}
    \alpha\leq \frac{1+\frac{4-\pi}{\pi}\sqrt{2}\Delta n (1-p)^D}{1+\sqrt{2}\Delta n (1-p)^D}
\end{equation}

Since $0<\frac{4-\pi}{\pi}<1$, we obtain the following sufficient condition on $\alpha$, concluding our proof.

\begin{equation}\label{equationAlphaQubo}
    \alpha\leq 1-\frac{2\pi-4}{\pi}\sqrt{2}\Delta n (1-p)^D
\end{equation}

\end{proof}

To the contrary of the Max-Cut case, our randomized rounding algorithm doesn't allow us to sample from every QUBO circuits at fixed depth, and we require $\log n$ depth to sample with high fidelity from the quantum circuit. However, without assuming more structure on the variance-covariance matrix, obtaining an approximation algorithm independent of $n$ would require further work on our approximation algorithm. Indeed, let us take the case where $\alpha > \frac{2}{\pi}$. We consider for all $i \neq j$ $\Sigma_{ij}>0$ such that $\frac{2}{\pi}\arcsin\Sigma_{ij}-\alpha\Sigma_{ij}=-\epsilon$ for some $\epsilon>0$. The matrix $\Sigma$ thus constructed can be positive semidefinite, and we get by denoting $J$ the matrix full of ones.

\begin{equation}
    M_\alpha=(1-\alpha)I-\epsilon(J-I)
\end{equation}

The eigenvalues of $\epsilon(J-I)$ are easily seen to be $(n-1)\epsilon$ with multiplicity one and $-\epsilon$ with multiplicity $(n-1)$. The matrices obviously commute such that the eigenvalues of $M_\alpha$ are $\lambda_1=(1-\alpha)+\epsilon$ with multiplicity $(n-1)$ and $\lambda_2=(1-\alpha)-(n-1)\epsilon$ with multiplicity one. Therefore, no matter how small $\epsilon$, for $n>1+\frac{1-\alpha}{\epsilon}$, $M_\alpha\notin S_n^+$.\\

Since Eq. (\ref{equationAlphaQubo}) requires $\log n$ depth to have a rigorous bound on how close the samples are to the quantum circuit, one may wonder how our approach differs from sampling the uniform distribution. Indeed, under depolarizing noise, the output of the quantum circuit is $\epsilon$-close to the maximally mixed state $\sigma=I/2^n$ after $\log(n/\epsilon)$ depth. Sampling from the uniform distribution would yield the same asymptotic bound; however, our approach offers several key advantages. First, our algorithm provides a standardized method for generating samples regardless of the noise regime, with a worst-case guarantee that improves as the circuit depth exceeds a certain threshold. Moreover, even at a fixed depth, the algorithm produces samples that closely resemble those generated by the quantum circuit, as demonstrated in Figure \ref{fig:boxplots} of Section \ref{sectionNumerics}. This occurs because the worst-case covariance matrix for a given QUBO problem $A$ can differ significantly from the one prepared by the quantum circuit. Finally, our sampling algorithm focuses on producing samples that resemble those of the quantum circuit. In this sense, our sampling approach differs from other works trying to simply beat the quantum algorithm \cite{StilckFrana2021}.

\section{Computing the variance-covariance matrix}\label{sectionCov}

The main input we use throughout this paper for our algorithms is the variance-covariance matrix of the circuit, as defined in Def. \ref{defSigmaMu}. Being able to compute this quantity in practice is therefore of upmost importance as it constitutes the potential bottleneck of our randomized rounding algorithm.  Even though computing the variance-covariance matrix is a challenging task, being bounded-error quantum polynomial time-complete \cite{Nielsen2012}, we outline several scenarios in which computing this matrix is feasible, yet obtaining samples remains out of reach.

\subsection{Noisy expectation values}

The simulation of noisy quantum devices has become an active area of research, with significant attention focused on computing noisy expectation values of Pauli observables \cite{fontana2023classical,Martinez2025}. Many of these simulations leverage techniques such as Pauli backpropagation, where the observable of interest, for instance $Z_i Z_j$, is backpropagated through the circuit in the Heisenberg picture. By incorporating the effects of depolarizing noise and employing smart truncation methods, it has been demonstrated that recovering noisy expectation values is feasible in \textit{polynomial} time \cite{fontana2023classical}, with inverse polynomial precision error.

The convergence guarantees for such methods often require randomness in the quantum circuits considered. For example, \cite{fontana2023classical} and \cite{Martinez2025} showed that for circuits composed of Clifford layers and single-qubit $R_z$ rotations, it is possible to recover the expectation value of any Pauli observable in polynomial time, with convergence guarantees \textit{in average over rotation parameters}. Similarly, \cite{schuster2024} proposed an algorithm for fixed quantum circuits, with guarantees \textit{in average over some input states set}. Recent advances have extended these results to a wider class of circuits and more general noise models \cite{Martinez2025}, including nonunital noise, demonstrating that approximating expectation values classically can be efficiently achieved for almost any noise channel and broad circuit classes.

Beyond theoretical guarantees, these simulation techniques have been shown to be highly effective in practice, accurately simulating noisy and noiseless quantum systems with high precision \cite{rudolph2023, angrisani2024, angrisani2025}. However, these methods are primarily designed to recover expectation values from quantum circuits, and adapting them to generate samples is not straightforward. Currently, classical sampling algorithms for noisy circuits rely heavily on conditions such as anti-concentration of the output distribution \cite{schuster2024}, specific circuit classes (e.g., Clifford or IQP circuits) \cite{nelson2024}, or sufficiently deep circuits where simulation becomes trivial \cite{DePalma2021}. Unfortunately, none of these conditions are applicable to optimization circuits such as QAOA. This highlights the importance of our framework as a practical tool for obtaining samples from noisy quantum circuits designed for combinatorial optimization.

\subsection{Noiseless expectation values}

Another scenario where expectation values can be efficiently computed but sampling from the quantum circuit remains challenging is shallow quantum circuits with local observables. Even in the noiseless setting, when the circuit is sufficiently shallow, expectation values can be computed in polynomial time on a classical computer using lightcone techniques~\cite{Bravyi2021, Wild2023}. However, sampling from such circuits is still known to be computationally difficult.

Quantum error mitigation presents another setup where noiseless expectation values can be recovered, albeit often at an exponential cost, but samples cannot be directly retrieved \cite{Cai2023}. Quantum error mitigation is proposed as a strategy to address the unavoidable errors that arise in the early stages of quantum computing. This technique involves classical post-processing of results obtained from various quantum circuits and requires minimal or no additional quantum resources \cite{Temme2017}. Unlike fault-tolerant quantum computing approaches, which typically involve substantial overhead \cite{Steane1996, Kitaev1997}, error mitigation offers a practical way to reduce noise in modest quantum computational setups. In practice, several error mitigation procedures exist and are already implemented. This includes virtual distillation \cite{Huggins2021}, Clifford data regression \cite{Czarnik2021}, zero-noise extrapolation \cite{Temme2017}, and probabilitic error cancellation \cite{Piveteau2022}.

Note that in our setting, employing the sampling error mitigation protocol from \cite{Liu2025} is not suitable. Indeed, let us clarify what the method of Liu and Cai achieves, which will then show why it is not suitable for the task at hand. Given measurement outcome bitstrings $z_1,\ldots,z_m\in\{0,1\}^n$ obtained from the noisy quantum device, the method of \cite{Liu2025} provides error-mitigated estimates of the corresponding noiseless outcome probabilities $p_{z_1},\ldots,p_{z_m}$, and does so more efficiently than naive error mitigation. This is useful in applications such as quantum phase estimation. However, the method only yields corrected probabilities for bitstrings that have already been observed. From a combinatorial optimization perspective, such corrected probabilities are not essential: we can simply evaluate the objective value of each observed bitstring and output the one achieving the lowest energy. Therefore, these error-mitigation schemes are not suitable for the problem at hand. What would be relevant instead are methods that generate new bitstrings with improved energy, starting from noisy samples.

All these error mitigation schemes require $\mathcal{O}(p^{-\Omega(D)})$ samples, where $D$ represents depth of the quantum circuit. Using these error mitigation schemes in practice therefore requires taking a close look at the circuit, as the complexity can quickly explode. Fixed depth quantum circuits, such as QAOA, are good candidates for error mitigation schemes.

Even though the recovery ratio of our sampling scheme is theoretically worse in these scenarios, we show numerically in Section \ref{sectionNumerics} that our sampling method performs well in practice and offers an alternative to obtain ``good" samples from the quantum circuit.

\subsection{Projection onto PSD cone}

In the following, we will suppose that we are able to approximate the expectation value of our observables of interest, i.e. the circuit mean vector and variance-covariance matrix as defined in Def.~\ref{defSigmaMu}. We give the following definition for this approximation.

\begin{defi}[Computing expectation values]\label{defExpValues} Let $O_1,\ldots,O_m$ be a set of $m$ observables such that $\norm{O_i}\leq1$. We say that we are able to compute the expectation values of these $m$ observables if there exists an $(\eta, \delta)$-procedure that takes as input a description of the quantum circuit $\mathcal{C}$ and of the noise channel $\mathcal{N}$ acting on $\mathcal{C}$ and outputs estimates $\tilde{\langle O_i\rangle}$ of $\langle O_i \rangle=\Tr([\mathcal{C}]_\mathcal{N}O_i)$ such that,
    \begin{equation}
        \quad \mathbb{P}[|\langle O_i\rangle-\tilde{\langle O_i\rangle }|\leq\eta|,\quad \forall \space 1\leq i \leq m]\geq 1-\delta
    \end{equation}
\end{defi}

This definition covers all the scenarios considered above. For broad classes of noise channels $\mathcal{N}$ and quantum circuits, the methods in \cite{Martinez2025} enable recovering noisy expectation values in polynomial time. Classical simulations are also included within this definition by setting the noise channel to the identity. Similarly, by providing multiple copies of the noisy quantum circuit as input, error mitigation schemes can recover noiseless expectation values of the quantum circuit \cite{Cai2023}.

Algorithm \ref{samplingAbs} requires sampling from a Gaussian distribution defined by the circuit mean vector and variance-covariance matrix as described in Def. \ref{defSigmaMu}. However, if these quantities are approximated as outlined in Def. \ref{defExpValues}, the resulting covariance matrix $\tilde{\Sigma}$ may not be positive semidefinite. To address this issue, we propose projecting $\tilde{\Sigma}$ onto the cone of positive semidefinite matrices and derive the following result:

\begin{prop}[Projection onto $S_n^+$]\label{prop:projectionSn}
 Let $\tilde{\Sigma}$ be an unbiased approximation of the variance-covariance matrix $\Sigma$ with precision $\eta$, as defined in Def. \ref{defExpValues}. Suppose additionally that the errors made on each of the estimated coefficients are independent. It is possible to recover, in polynomial time, a matrix $P(\tilde{\Sigma}) \in S_n^+$ such that there exists a constant $C \geq 0$ satisfying, for any $A \in S_n$: \begin{equation} \tr(A(P(\tilde{\Sigma}) - \Sigma)) \leq C \eta \norm{A}_F (n + \sqrt{n} t), \end{equation} with probability at least $1 - 2\exp(-t^2)$. If $A$ is $\Delta$-regular, $\norm{A}_F \sim \sqrt{\Delta n}$. 
\end{prop}

\begin{proof}
The covariance matrix coefficients can be treated as independent bounded random variables satisfying: 
\begin{equation}
\forall (i,j) \in E, \quad |\Sigma_{ij} - \tilde{\Sigma}_{ij}| \leq \eta. 
\end{equation}

The entries of $\tilde{\Sigma}$ are sub-Gaussian random variables centered around the entries of $\Sigma$. From known results \cite{Vershynin2018}, it follows that there exists a constant $C \geq 0$ such that, for all $t > 0$, $\norm{\tilde{\Sigma} - \Sigma} \leq C \eta (\sqrt{n} + t)$ with probability at least $1 - 2\exp(-t^2)$. Note that while we assumed an unbiased approximation to ensure centered errors, the main result relies solely on the validity of this operator norm bound.
However, the matrix $\tilde{\Sigma}$ may not be positive semidefinite, making Gaussian sampling infeasible. To resolve this, we project $\tilde{\Sigma}$ onto the cone of positive semidefinite matrices $S_n^+$. This projection is achieved by computing the eigenvalues $(\lambda_1, \ldots, \lambda_n)$ of $\tilde{\Sigma}$ and replacing them with $\max(0, \lambda_i)$. The resulting matrix, denoted $P(\tilde{\Sigma})$, satisfies the properties of being in $S_n^+$ and minimizing both the Frobenius norm and the spectral norm \cite{Boyd2004}.

We then analyze the deviation between the projected variance-covariance matrix and the original matrix. With probability at least $1 - 2\exp(-t^2)$,

\begin{align}
     \|P(\tilde{\Sigma})-\tilde{\Sigma}\|_F\leq  \norm{\Sigma-\tilde{\Sigma}}_F\leq \sqrt{n}\norm{\tilde{\Sigma}-\Sigma}\leq  C\eta (n +\sqrt{n}t)
\end{align}

where the inequalities follow from properties of the projection and relations between the operator norm and the Frobenius norm. Since our cost function can be expressed as $\Tr(A \Sigma)$ for some $A \in S_n$, we use the matrix Hölder inequality to derive: 

\begin{align}
    \tr(A(P(\tilde{\Sigma})-\Sigma))=\tr(A(P(\tilde{\Sigma})-\tilde{\Sigma}))+\tr(A(\tilde{\Sigma}-\Sigma))\leq \\\|A\|_F(\|P(\tilde{\Sigma})-\tilde{\Sigma}\|_F+\|\Sigma-\tilde{\Sigma}\|_F)\leq 2C\eta\|A\|_F(n+\sqrt{n}t)
\end{align}

\noindent with probability at least $1-2\exp(-t^2)$. It is worth noting that if $A$ is $\Delta$-regular, $\|A\|_F\sim\sqrt{\Delta n}$.
\end{proof}

Prop. \ref{prop:projectionSn} indicates that to ensure the approximation error does not significantly affect the samples obtained through our randomized rounding approach, inverse polynomial precision $\propto n^{-3/2}$ is required for the estimated coefficients of $\Sigma$. Conveniently, this level of precision can be achieved in polynomial time for noisy quantum circuits using Pauli backpropagation techniques \cite{angrisani2025, Martinez2025}.

\subsection{End-to-end framework for sampling from noisy circuits}

In this section, we propose an end-to-end analysis on how to sample classically from noisy circuits made for combinatorial optimization. This analysis builds upon recent results for computing efficiently the expectation values of Pauli observables of noisy quantum circuits \cite{fontana2023classical,Martinez2025}. For most noise models, including any combination of depolarizing, dephasing, or amplitude damping noise, a polynomial-time algorithm was derived to compute these noisy expectation values, with up to inverse polynomial precision. In the following, we will consider parametrized quantum circuits made of $D$ alternating layers of single-qubit rotations $\mathcal{R}^{(q_i)}_z$ and Clifford operations $\mathcal{C}_i$, such that,

\begin{equation}\label{eq:unitaryLOWESA}
\mathcal{U_{\theta}}=\bigg(\bigcirc_{i=1}^D\mathcal{C}_i\circ \mathcal{R}^{(q_i)}_z(\theta_i)\circ\mathcal{N}^{\otimes n}\bigg)\circ\mathcal{C}_0,
\end{equation}

Common variational quantum circuits made for optimization purposes, such as QAOA or VQE, fit this description.

\begin{lem}[Adapted from \cite{fontana2023classical} and \cite{Martinez2025}] Let $\mathcal{U_{\theta}}$ be a quantum circuit made of $D$ layers of parametrized rotations and fixed Clifford gates as defined in Eq. \ref{eq:unitaryLOWESA}. Let $P$ be a Pauli observable, for us $P=Z_iZ_j$, which expectation value is of interest. 

\begin{itemize}
    \item \textbf{Depolarizing noise} (Def. \ref{defNoiseDepo}): If $\mathcal{U_{\theta}}$ is subject to depolarizing noise of strength $p$, then Pauli propagation algorithm of \cite{fontana2023classical} recovers the expectation value $\tr([\mathcal{U_{\theta}}]_{\mathcal{N}_{DP}}P)$ up to precision $t$ with probability at least $1-\epsilon/t$ in time $\mathcal{O}(n^2D(1/\epsilon)^{1/p})$, with the probability taken over the angles of the quantum circuit.
    \item \textbf{Amplitude damping noise} (Def. \ref{defNoiseAD}): Similarly, if $\mathcal{U_{\theta}}$ is subject to amplitude damping noise of strength $p$, then Pauli propagation algorithm of \cite{Martinez2025} recovers the expectation value $\tr([\mathcal{U_{\theta}}]_{\mathcal{N}_{AD}}P)$ up to precision $t$ with probability at least $1-\epsilon/t$ in time $\mathcal{O}(n^2D(1/\epsilon^2)^{1+1/p})$, with the probability taken over the angles of the quantum circuit.
\end{itemize}
\end{lem}\label{lem:PauliBackprop}

Note that these simulation algorithms can be extended to handle more general single-qubit noise models \cite{Martinez2025} and are not restricted to depolarizing or amplitude damping noise. Additionally, no assumptions are made on the locality of the observable or the geometry of the underlying quantum chip, which is particularly relevant for optimization tasks involving highly connected graphs.

For most noisy quantum circuits designed for optimization tasks, it is possible to recover expectation values up to inverse polynomial precision in polynomial time. However, these algorithms are not inherently designed to produce samples from the quantum circuit. To address this, we propose using simulation algorithms to classically recover the correlation matrix of the circuit and then applying our sampling scheme to this matrix. This provides an end-to-end framework for classically sampling from noisy quantum circuits designed for optimization. The results are summarized in the following theorem:

\begin{theorem}[Classically sampling from Max-Cut circuits under depolarizing noise] Let $\mathcal{G}=(V,E)$ be a graph with maximum degree $\Delta$ and uniform weights. Let $[\mathcal{U_{\theta}}]_{\mathcal{N}_{DP}}$ be the state prepared by the noisy quantum circuit of Eq. \ref{eq:unitaryLOWESA} with noise strength $p$. Denote $\varepsilon^2 = 2\sqrt{2}\Delta(1-p)^D$ the bound achieved by the average correlation coefficient. Recall that $g$ is the function $g: x\mapsto \frac{2\arccos x}{\pi(1-x)}$ introduced in Figure \ref{plotFunctiong}, such that $g(0)=1$. If $\varepsilon\leq 0.689$, then Pauli propagation algorithms \cite{fontana2023classical,Martinez2025} coupled with sampling Algorithm \ref{samplingAbs} produce samples such that with probability at least $1-\delta-2\exp(-n)\leq 1-2\delta$ as long as $n\in\Omega(\log1/\delta)$,
\begin{equation}
\begin{aligned}
    \alpha_{MC'}=\frac{\mathbb{E}_{z\sim\mathcal{A}(\mathcal{U_{\theta}})}[C(z)]}{\mathbb{E}_{z\sim\mathcal{M}([\mathcal{U_{\theta}}]_{\mathcal{N}_{DP}})}[C(z)]}&\geq(1-\varepsilon)g(-\varepsilon)+ \varepsilon\alpha_{GW}\\
    & \geq 1 - O(\varepsilon)
\end{aligned}
\end{equation}

in time $\mathcal{O}\bigg(\Delta n^3 D\big(\frac{\Delta^{3/2}n^{5/2}}{\delta}\big)^{1/p}\bigg)$, with the probability taken over the angles of the quantum circuit.

\end{theorem}

\begin{proof}
    The first step of the framework consists in approximating each coefficient of the correlation matrix using the Pauli backprogation algorithm of Lemma \ref{lem:PauliBackprop}, such that with failure probability $\delta$,

    \begin{equation}
        \quad \mathbb{P}[|\tilde{\Sigma}_{ij}-\Sigma_{ij}|\leq\eta|,\quad \forall (i,j)\in E]\geq 1-\delta
    \end{equation}
    
This can be done using a simple union bound. The number of edges being at most $\Delta n$, running the Pauli backpropagation of Lemma \ref{lem:PauliBackprop} with target precision $\epsilon= \delta\eta/\Delta n$ yields the required bound. The approximation of the correlation matrix computed, the second step consists in projecting it onto $S_n^+$ to then sample from the Gaussian distribution with matching moments. As shown in proposition \ref{prop:projectionSn}, this can be done with failure probability $2\exp(-t^2)$ and total precision on the cut function $C\eta\norm{A}_F(n+\sqrt{n}t)$, where $A$ corresponds to the Max-Cut instance as defined in Eq. (\ref{eq:MatrixMaxCut}). Picking $t=\sqrt{n}$ and using the fact that $\norm{A}_F\leq\sqrt{\Delta n}$, we obtain that with failure probability $\exp(-n)$,

\begin{align}
\tr(A(P(\tilde{\Sigma})-\Sigma))\leq\eta n^{3/2}\sqrt{\Delta}
\end{align}

To obtain inverse polynomial precision on the recovery ratio, it is only necessary to ensure that $ \tr(A(P(\tilde{\Sigma})-\Sigma))\in\mathcal{O}(1)$. This can be done by running the Pauli backpropagation algorithm with total precision $\epsilon\propto \delta/\Delta^{3/2} n^{5/2}$. This needs to be done for all edges of the graph. Since there are at most $\Delta n$ edges, the total runtime complexity is at most $\mathcal{O}\bigg(\Delta n^3 D\big(\frac{\Delta^{3/2}n^{5/2}}{\delta}\big)^{1/p}\bigg)$. Note that sampling the Gaussian distribution can be done efficiently in time linear in system size.
Finally, the bound on the recovery ratio is obtained directly through Th. \ref{theoremNoisyMaxcut} for circuits made of $D$ layers with noise strength $p$.
    
\end{proof}

Even though this polynomial-time algorithm might appear computationally expensive, the primary cost arises from estimating the correlation matrix, which has been shown to be highly effective in practice, even in low-noise and noiseless regimes \cite{rudolph2023,angrisani2024}. As a result, the proposed framework not only produces samples from the noisy circuit in polynomial time theoretically but also performs efficiently in practice.
The same framework can be extended to sample from circuits under more general noise models with similar runtime using the results of \cite{Martinez2025}. While no analytical bounds are currently available for the quality of samples in these regimes, we demonstrate numerically in Section \ref{sectionNumerics} that under amplitude damping noise, the samples generated by Algorithm \ref{samplingAbs} achieve an increasing recovery ratio as the circuit depth and noise strength grow.
\section{Numerics}\label{sectionNumerics}

\subsection{Quantum circuits considered}\label{sectionQAOA}

Since our randomized rounding algorithm can be implemented efficiently, we analyze its performance on real instances of Max-Cut and QUBO. The hardware utilized for these experiments will be IBM's 127-qubit quantum chips \cite{mckay2023benchmarkingquantumprocessorperformance}. By aligning the combinatorial optimization problems with the chip's architecture, we ensure minimal transpilation overhead. Our study encompasses both noiseless simulations and real hardware experiments, as both scenarios are relevant and extend beyond the analytical results presented in Sections \ref{sectionMC} and \ref{sectionQUBO}.

The quantum circuits we employ are instances of the well-known Quantum Approximate Optimization Algorithm (QAOA), which prepares the state:

\begin{equation}
    \ket{\psi(p,\beta,\gamma)} = \prod_{j=1}^p e^{-i\beta_j H_X} e^{-i\gamma_j H} \ket{+},
\end{equation}

where $H_X = -\sum_{j=1}^n X_j$ and $H$ is the 2-local Hamiltonian encoding the optimization problem. The parameter $p \in \mathbb{N}^*$ defines the depth of the QAOA circuit, while $(\beta, \gamma) \in \mathbb{R}^p \times \mathbb{R}^p$ are variational parameters to be optimized. For the graphs considered—specifically 3-regular graphs or graphs closely resembling the hexagonal architecture of the quantum chip—we use fixed angles $(\beta, \gamma)$ that were previously derived and shown to approximate the optimal parameters for such graphs \cite{Wurtz2021}.

\subsection{Numerical results}\label{sectionQAOA}

\begin{itemize}

\item \textbf{Beyond Worst-Case Recovery Ratios (Figure \ref{fig:boxplots}).} For this first experiment, we consider QAOA circuits executed on real hardware on graph instances with over 100 nodes. The graphs are constructed as follows: starting with the chip connectivity graph as a reference, every pair of nodes with a Hamming distance of two or less is connected. Edges are then randomly removed from this densely connected graph until the total number of edges matches that of the original chip connectivity graph. This construction ensures minimal transpilation cost and reduces the circuit depth overhead caused by additional swap gates. Additionally, we examine a specific instance of the chip connectivity graph itself, which further minimizes transpilation cost and circuit depth. 

For these graphs and their corresponding QAOA circuits ($p=1$ and $p=3$), we compute the covariance matrix on IBMQ hardware and apply our rounding procedure. The resulting recovery ratios are presented as box-plots in Figure \ref{fig:boxplots} (Hardware). For reference, we also generate noiseless covariance matrices using light-cone simulations via the \texttt{quimb} library \cite{gray2018quimb}. The rounding is then performed identically to obtain samples.

In this simulated noiseless regime, where the covariance matrix of the noiseless circuit is retrieved, Figure \ref{fig:boxplots} demonstrates that for both QUBO and Max-Cut instances, our approach consistently exceeds the theoretical worst-case recovery ratios provided in Eq. (\ref{eq:maxcut_worstcase}) and (\ref{eq:QUBO_worstcase}).

For Max-Cut, the worst-case recovery ratio, $\alpha_{GW} \approx 0.878$, is only achieved under very specific conditions where, for each $(i,j) \in E$, the covariance $\Sigma_{ij} = \langle Z_i Z_j \rangle \approx -0.689$. Unless the graph and QAOA circuit are meticulously engineered to produce outputs with this covariance, the recovery ratio is expected to perform better than the theoretical bound. 

Since we analyze the performance of the algorithm in the intermediate noise regime (i.e. the noise is not strong enough to approximate the output by the maximally mixed state but not weak enough to approximate it by the noiseless computation), the performance of the algorithm is influenced in a subtle way by competing effects. Adding additional gates increases the noise in the circuit, which acts to depolarize the state and wash out correlations. However, it also biases the output toward good bitstrings, where adjacent nodes are negatively correlated to maximize the cut size.

Because of these competing effects, when increasing the depth of QAOA, the recovery ratio of our procedure may worsen, which can be understood by examining Figure \ref{plotFunctiong}. Smaller values of $\Sigma_{ij}$ tend to place adjacent edges in different sets, thereby improving the cut value. Although an increase in depth results in stronger noise effects, the improvement in performance due to the larger QAOA depth $p=3$ outweighs the losses caused by the noise. Unfortunately, this doesn't hold for greater values of $p$ as the noise kicks-in exponentially fast.

Similarly, the worst-case bound of $2/\pi$ for QUBO is realized for a specific covariance matrix for a given graph. In practice, however, our method significantly outperforms this theoretical bound.

\begin{figure}[H]
    \centering
    \includegraphics[width=0.51\textwidth]{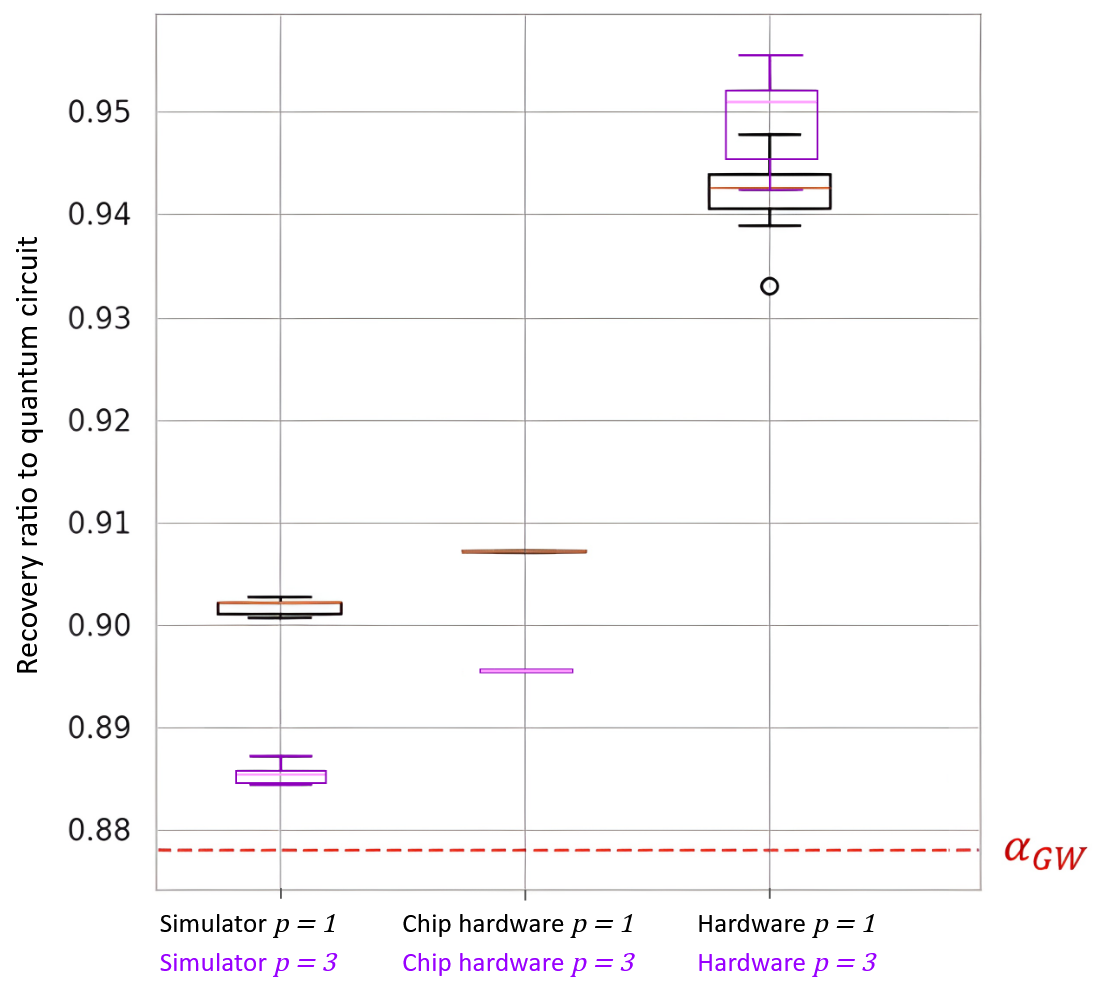}
    \includegraphics[width=0.48\textwidth]{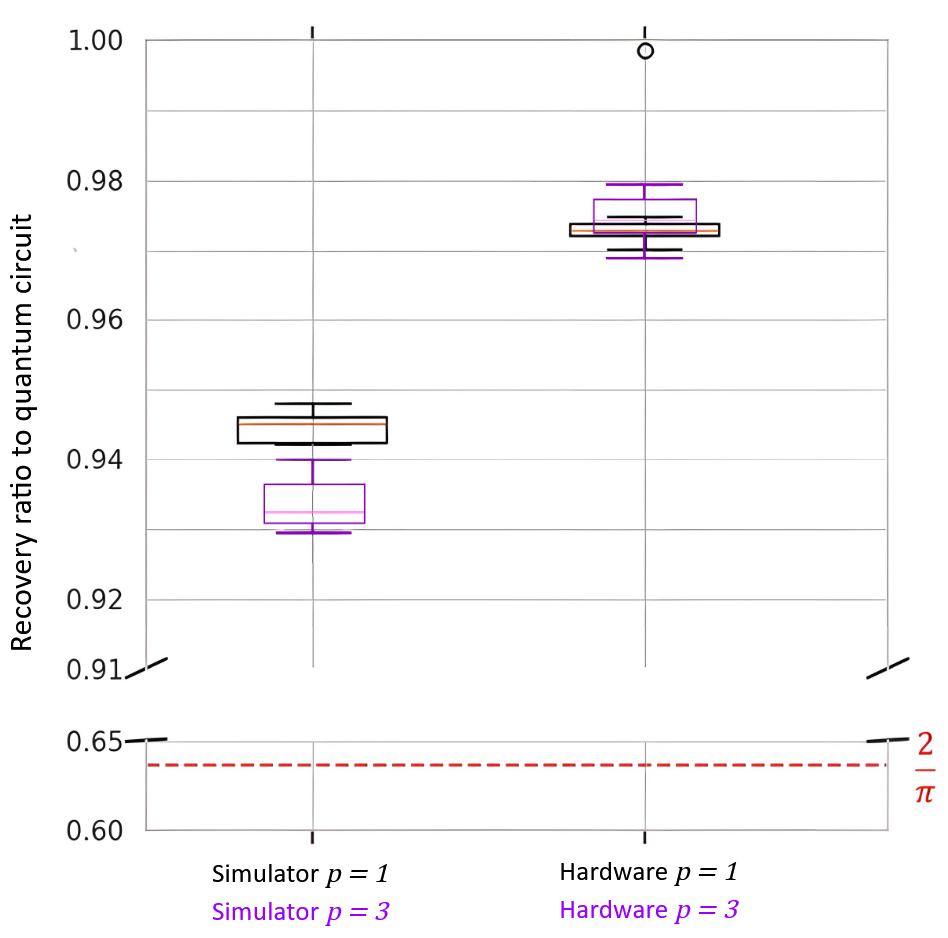}
    \caption{Box-plot of the recovery ratios achieved by the sampling Algorithm \ref{samplingAbs} on random instances of Max-Cut (left) and QUBO (right), as defined in Th. \ref{theoremNoisyMaxcut} and Th. \ref{theoremNoisyQUBO}. The graphs studied are random graphs with more than 100 nodes, where edges between qubits exist only if the qubits are separated by at most two edges in the original chip graph. Additionally, a single instance of the chip connectivity graph (Chip Hardware) is considered. The covariance matrices of the noiseless circuits (Simulator) are obtained using lightcone techniques to compute the cut produced by the quantum circuit, while hardware results (Hardware) correspond to quantum circuits executed directly on the quantum chip.
    \label{fig:boxplots}}
\end{figure}

\item \textbf{Sampling from the Noisy Quantum Circuit Beyond Expectation Values (Figure \ref{fig:Distributions}).} Th. \ref{theoremNoisyMaxcut} and \ref{theoremNoisyQUBO} demonstrate that, under depolarizing noise, after a fixed number of gates (or logarithmic depth for QUBO), the samples produced by our procedure match, in expectation, those generated by the noisy quantum circuit. However, recent research has shifted focus from optimizing the expectation value of the objective function produced by the quantum circuit to analyzing the tail of the distribution of the objective function \cite{Barron2024}. This approach involves sampling the quantum computer multiple times and retaining only the top $\alpha$ fraction of samples, where $\alpha$ is a predetermined parameter. This quantity, often referred to as the Conditional Value-at-Risk (CVaR), has been shown to recover noiseless samples for certain well-chosen values of $\alpha$, which depend on the noise level in the circuit. Consequently, an important question arises: how does our sampling algorithm perform on the tail of the distribution? Do the tails of the distributions align, or is the correspondence limited to expectation values?

For a single Max-Cut instance with 40 nodes, corresponding to the 3-regular graph of \cite{Barron2024}, we show that this correspondence extends beyond the average cut value by analyzing multiple samples obtained from QAOA circuits run on real quantum hardware. When comparing the distribution of cuts generated by our rounding procedure to those obtained by directly sampling the quantum circuit, the entire distributions appear to align, not just their expectation values. Remarkably, this agreement holds even for the tail of the cut distribution, despite expectations that noiseless samples would dominate in that region \cite{Barron2024}.

This result, visualized in Figure \ref{fig:Distributions}, underscores that our sampling algorithm provides a robust alternative to direct sampling from the quantum circuit for solving optimization problems. Furthermore, it suggests that the strategy introduced in \cite{Barron2024}, which involves sampling the quantum circuit multiple times to recover noiseless samples, could be replaced by sampling the Gaussian distribution and performing randomized rounding repeatedly, achieving similar results.

Note that even though the circuit is affected by noise and performs close to uniform sampling, the state prepared is not the maximally mixed state. In fact, the minimum correlation among some edges remains significantly far from zero (e.g., $-0.19$ in (a) and $-0.11$ in (b)). Consequently, the trace distance between the prepared state and the maximally mixed state is large. However, our technique only requires that the correlation matrix be close to that of the maximally mixed state \emph{on average}, which explains why the sampling algorithm performs well.

\begin{figure}[H]
    \centering
    \text{(a)} 
    \begin{minipage}{0.45\textwidth}
        \centering
        \includegraphics[width=\textwidth]{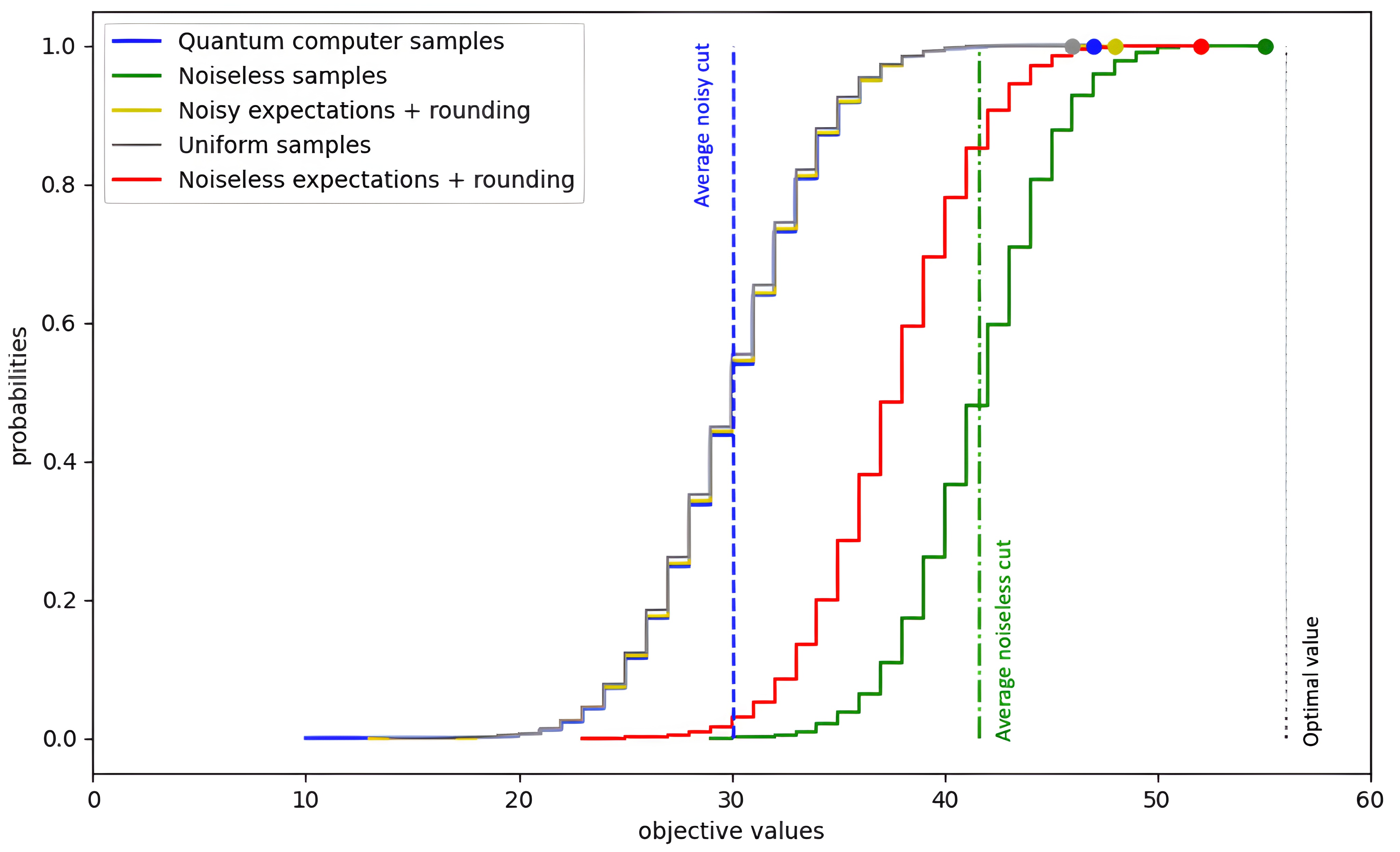} 
    \end{minipage}
    \begin{minipage}{0.45\textwidth}
        \centering
        \includegraphics[width=\textwidth]{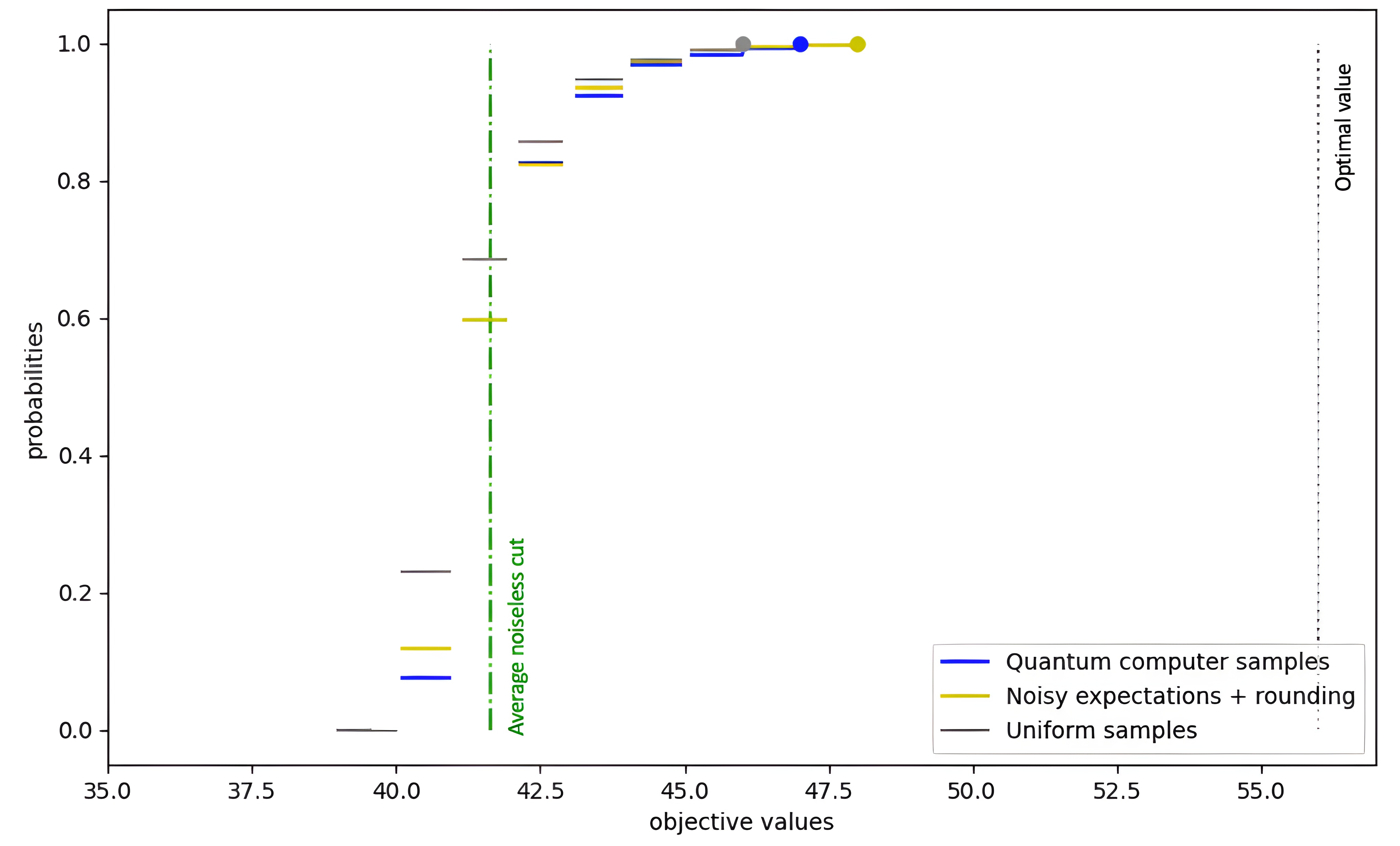} 
    \end{minipage}

    \text{(b)} 
    \begin{minipage}{0.45\textwidth}
        \centering
        \includegraphics[width=\textwidth]{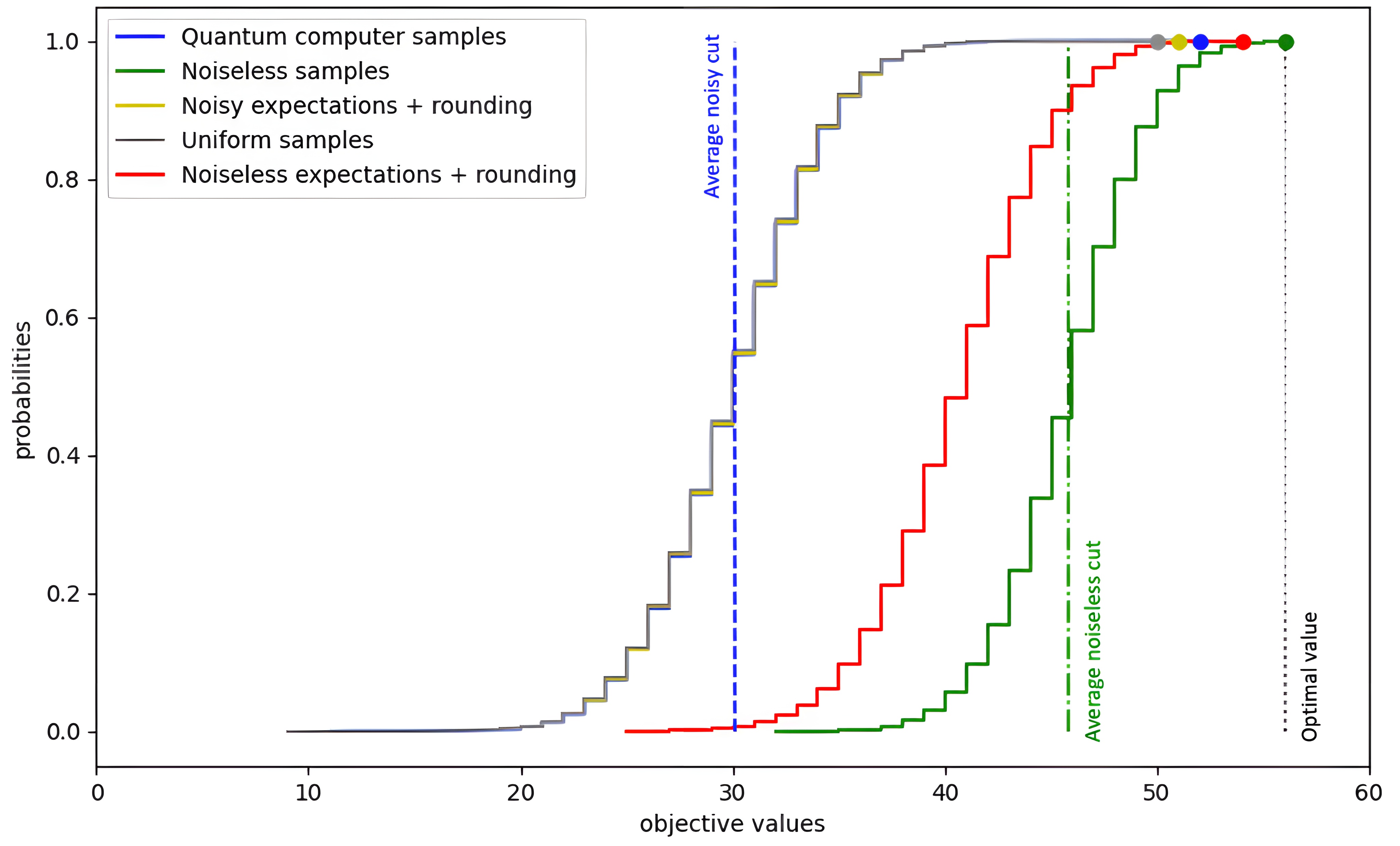} 
    \end{minipage}
    \begin{minipage}{0.45\textwidth}
        \centering
        \includegraphics[width=\textwidth]{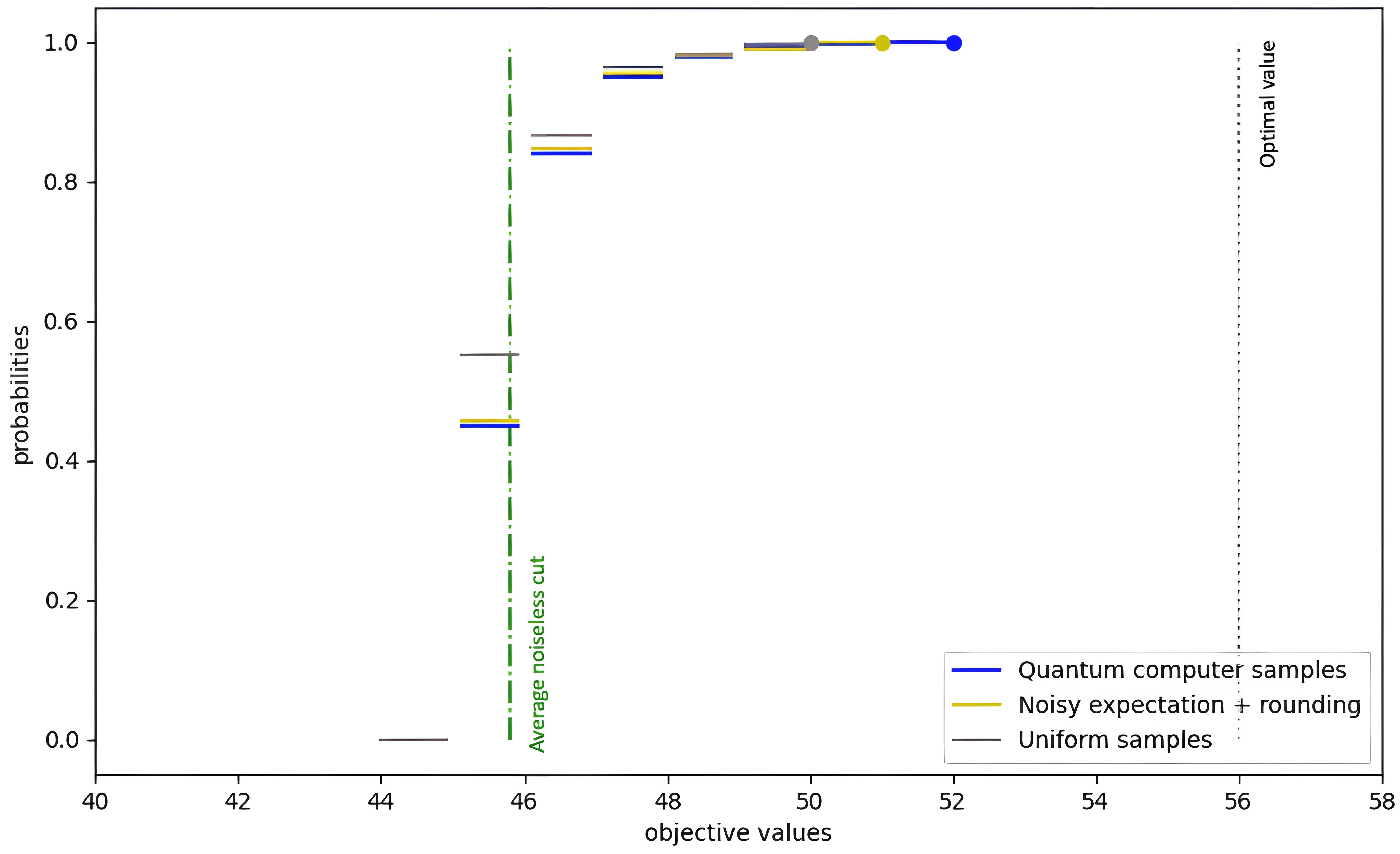} 
    \end{minipage}

    \caption{Cumulative distribution of cut values achieved by different sampling methods on a 40-qubit QAOA circuit, for $p=1$ (a) and $p=2$ (b). The full distribution (left) and the best $\alpha$ fraction of samples for a carefully chosen $\alpha$ (right) are shown \cite{Barron2024}. The noiseless samples (green) are obtained using PEPS tensor network simulations via the \texttt{quimb} library \cite{gray2018quimb}, while the noisy samples (blue) are directly sampled from the quantum device. Algorithm \ref{samplingAbs} is applied to both the noisy expectation values (yellow) and the noiseless expectation values (red) to generate samples. The quantum device is sampled 100,000 times in (a) and 10,000,000 times in (b) to ensure that the tail of the distributions theoretically matches the noiseless expectation values. The Gaussian distribution and the rounding is performed as many times for comparison, and the uniform distribution (grey) is sampled for comparison.}
    \label{fig:Distributions}
\end{figure}

\item \textbf{Beyond Unital Noise (Figure \ref{fig:Nonunital_approx}).} While Th. \ref{theoremNoisyMaxcut} and \ref{theoremNoisyQUBO} are derived under depolarizing noise, it is natural to question whether other types of noise, such as amplitude damping, yield similar results for the samples generated by our framework. 

Intuitively, in the regime of high noise rates, we anticipate the covariance matrix to converge to $\Sigma = \mathbf{1} \cdot \mathbf{1}^T$. This convergence would result in an recovery ratio of 1 for QUBO and an recovery ratio greater than 1 for Max-Cut. We confirm that this intuition holds across varying noise strengths and circuit depths by simulating additional amplitude-damping noise on a QAOA circuit made for a 3-regular graph with 16 nodes. The noise is added after each gate of the quantum circuit, and the expectation values computed using \texttt{qiskit}.

The results, presented in Figure \ref{fig:Nonunital_approx}, illustrate this behavior and provide a comparative analysis with depolarizing noise. 
\end{itemize}

\begin{figure}[H]
\centering
\includegraphics[width=0.48\textwidth]{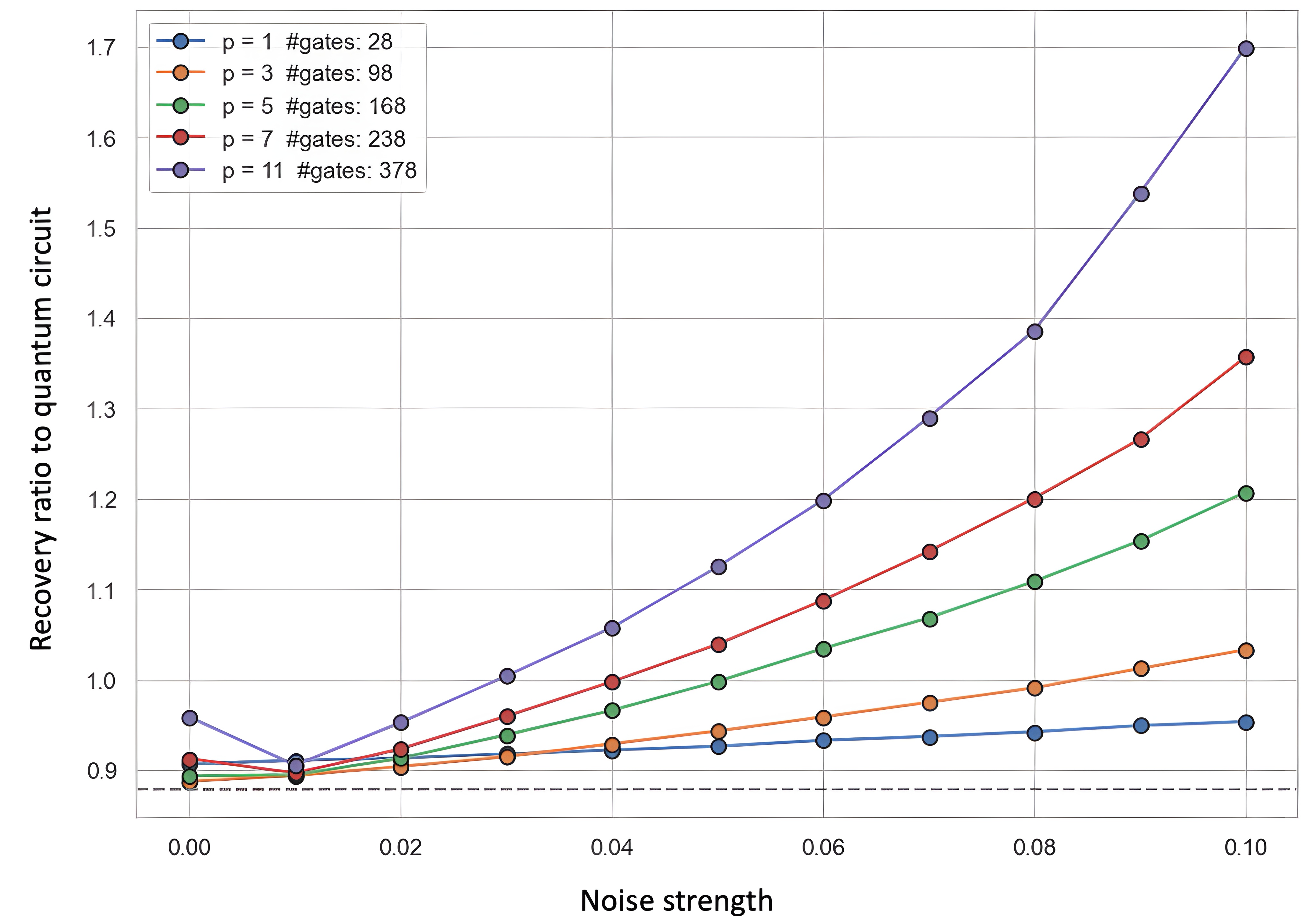}
\includegraphics[width=0.48\textwidth]{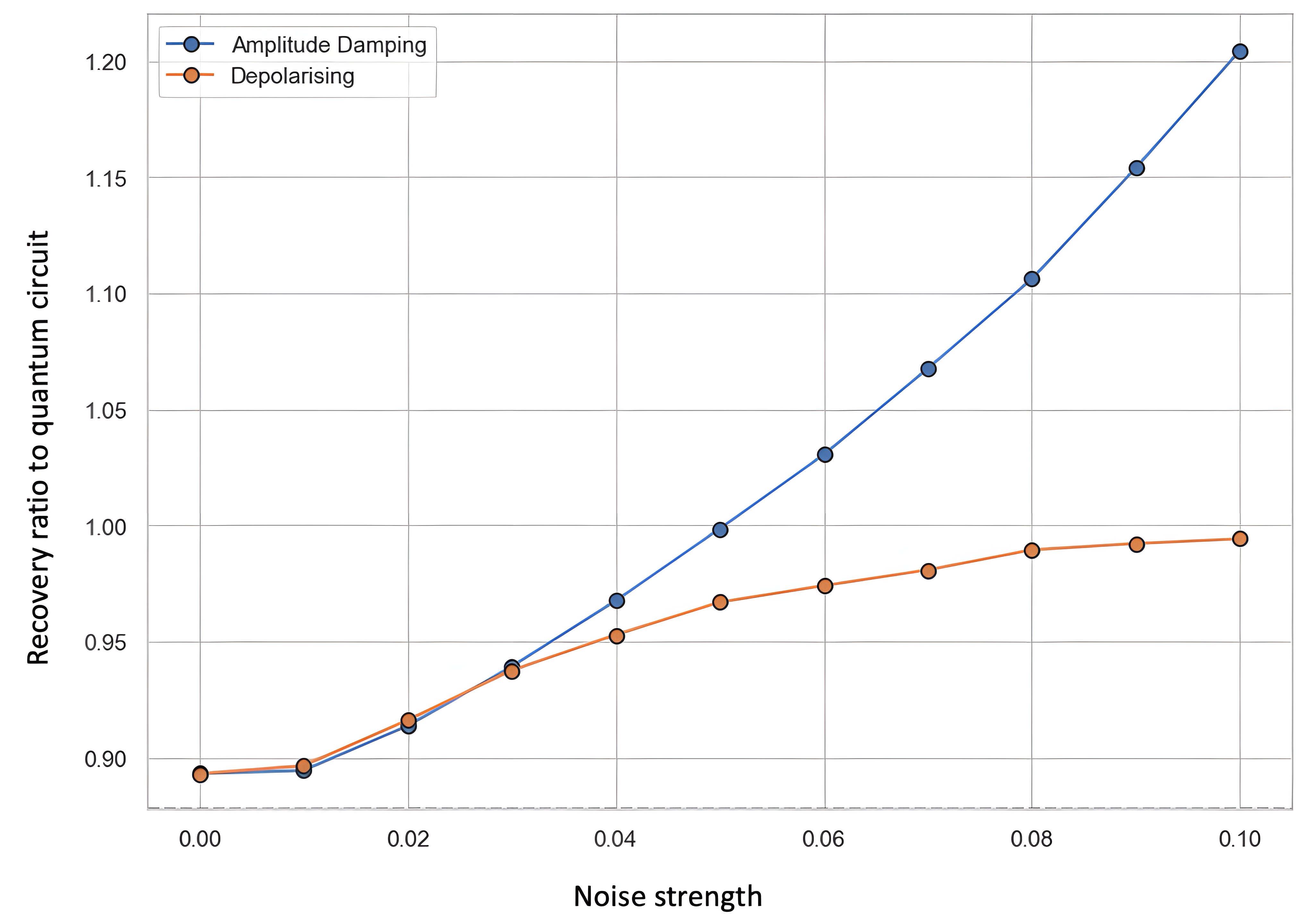}
\caption{Recovery ratios achieved by the sampling Algorithm \ref{samplingAbs} on QAOA circuits for Max-Cut under amplitude damping noise for different QAOA depths $p$ and noise strengths. As the QAOA depth increases, the average number of gates in the lightcone of the observables—and consequently the total noise—also increases. The results are compared to those obtained under depolarizing noise. Both types of noise exhibit a similar effect on the recovery ratio, which improves as the noise increases, enhancing the quality of the samples produced by the randomized rounding approach.}
\label{fig:Nonunital_approx}  
\end{figure}

\section*{Conclusion and outlook}
We have introduced a simple but powerful classical surrogate for sampling the output of noisy quantum circuits that target combinatorial-optimization tasks. By showing that Gaussian randomized rounding applied to the two-qubit marginals of any depth--$D$ circuit under local depolarizing noise with strength $p$ yields samples whose expected cost is at most an $O\bigl((1-p)^D)$ fraction away from that of the noisy device, we clarify where near-term hardware might—and might not—outperform classical algorithms. Experiments on IBMQ processors and large-scale simulations confirm that our rounding procedure approximately reproduces the entire energy distribution of noisy QAOA states, not just the mean, and that the agreement improves as noise or depth grows.
Coupled with recent Pauli-backpropagation techniques for efficiently estimating noisy marginals~\cite{fontana2023classical,Martinez2025}, our method becomes an end-to-end classical sampler that operates in polynomial time for a broad class of variational circuits. Thus, our results show that sufficiently noisy quantum circuits for combinatorial optimization can be sampled from. By bridging rigorous theory, practical simulations, and hardware data, this work helps demarcate the frontier beyond which genuine quantum advantage must reside for noisy quantum optmizers. 

Several avenues merit further investigation. For instance, adapting the sampler to higher-order optimization problems or to qudit systems could illuminate the limits of noisy quantum hardware beyond QUBOs. Finally, integrating our method with advanced error-mitigation pipelines may clarify how close NISQ processors already are to the thresholds where classical surrogates cease to be competitive.

\section{Acknowledgements} The authors thank Stefan W\"orner and Daniel J. Egger for useful feedback and insightful discussions on randomized rounding and use of noisy quantum computers for combinatorial optimization. Moreover, we would like to thank Zo\"e Holmes and Armando Angrisani for interesting discussions.
O.F. acknowledges financial support from the European Research Council (ERC Grant AlgoQIP, Agreement No. 851716) and a government grant managed by the Agence Nationale de la Recherche under the Plan France 2030 with the reference ANR-22-PETQ-0007. D.S.F. acknowledges financial support from the Novo Nordisk
Foundation (Grant No. NNF20OC0059939 Quantum for Life) and by the ERC grant GIFNEQ 101163938. This project was funded within the QuantERA II Programme that has received funding from the EU’s H2020 research and innovation programme under the GA No 101017733.

This work is funded by the European Union. Views and opinions expressed are however those of the author(s) only and do not necessarily reflect those of the European Union or the European Research Council Executive Agency. Neither the European Union nor the granting authority can be held responsible for them. 

\bibliographystyle{alphaurl}
\bibliography{bibliography}

\newpage
\appendix
\section{Transportation cost inequality}\label{appendixTC}
Let us first introduce the notations and definitions necessary to proving Th. \ref{theoremBoundCov}.
Recall that the operator norm is denoted by $\|O\|=\sup_\rho |\tr(O\rho)|$. We can similarly introduce an auxiliary Lipschitz norm on the operator $O$, telling us how well it can distinguish two states differing by a qubit only \cite{DePalma2021}.

\begin{equation}
    \|O\|_{L;p}=\max_{1\leq i\leq n}\sup_{\tr_i[\rho]=\tr_i[\sigma]}|\tr(O(\rho-\sigma))|
\end{equation}

In the case where $O$ is a diagonal operator, this norm serves as a Lipschitz constant. We further introduce the Wasserstein distance of order 1 between quantum states $\rho$ and $\sigma$ as :

\begin{equation}
W_1(\rho,\sigma)=\sup_{\|O\|_{L;p}\leq 1}\tr(O(\rho-\sigma))
\end{equation}

The transportation cost inequality \cite{DePalma2021} states that $\rho$ satisfies TC with constant $\beta>0$ if: 

\begin{equation}
 W_1(\rho,\sigma)^2\leq \frac{n}{2\beta}D\left(\rho\middle\|\sigma\right)
\end{equation}

In particular, it is known that for $\sigma=I/2^n$, 

\begin{equation}
 W_1(\rho,\frac{I}{2^n})^2\leq \frac{n}{2}D\left(\rho\middle\|\sigma\right)
\end{equation}

We can further compute in the case of local depolarizing noise of strength $p$ represented by the channel $\mathcal{N}_{DP}$ acting on our $D$ layer quantum circuit $\mathcal{C}$.

\begin{equation}
\begin{aligned}
    W_1([\mathcal{C}]_{\mathcal{N}_{DP}},\frac{I}{2^n})^2&\leq \frac{n}{2}(1-p)^{2D}D\left(\rho\middle\|\frac{I}{2^n}\right)\\
    &\leq \frac{n^2}{2}(1-p)^{2D}
\end{aligned}
\end{equation}

Consider a graph $\mathcal{G}=(V,E)$ representing our optimization problem, and denote $\Delta$ the maximum degree of $\mathcal{G}$. We also consider the Hamiltonian $H=\sum_{(i,j)\in E} s_{ij}Z_iZ_j$, where the coefficients $s_{ij}$ are such that $\abs{s_{ij}}=1$. We can compute how far the state prepared by the noisy quantum circuit is from the maximally mixed state in term of energy under the Hamiltonian $H$,

\begin{equation}
\begin{aligned}
\abs{\tr(H([\mathcal{C}]_{\mathcal{N}_{DP}}-\frac{I}{2^n}))}
& \leq \|H\|_{L;p}W_1([\mathcal{C}]_{\mathcal{N}_{DP}},\frac{I}{2^n})\\
& \leq \sqrt{2}\Delta n (1-p)^D
\end{aligned}
\end{equation}

On the other hand we get that,

\begin{equation}
    \abs{\tr(H([\mathcal{C}]_{\mathcal{N}_{DP}}-\frac{I}{2^n}))}=\abs{\sum_{(i,j)\in E} s_{ij}\Tr([\mathcal{C}]_{\mathcal{N}_{DP}}Z_iZ_j)}
\end{equation}

and by picking $s_{ij}=\sign[\Tr([\mathcal{C}]_{\mathcal{N}_{DP}}Z_iZ_j)]$ we get the desired inequality,
\begin{equation}
    \sum_{(i,j)\in E}\abs{\Sigma_{ij}}\leq\sqrt{2}\Delta n (1-p)^D
\end{equation}

\end{document}